\documentclass[12pt,draftcls,onecolumn]{IEEEtran}

\usepackage{latexsym}
\usepackage{graphicx}
\usepackage{array}
\usepackage{amsmath}
\usepackage{amsfonts}
\usepackage{amssymb}
\usepackage{amsthm}
\usepackage{algpseudocode}
\usepackage{algorithm}

\newtheorem{thm}{Theorem}
\newtheorem{lemma}{Lemma}
\newtheorem{prop}{Proposition}

\newcommand{\bR} {\boldsymbol{R}}

\newcommand{{\diag}} {\mathrm{diag}}

\def\bal#1\eal{\begin{align}#1\end{align}}
\newcommand{\bp} {\begin{proof}}
\newcommand{\ep} {\end{proof}}

\newcommand{{\bRF}} {\right\}}

\newcommand{\bi}{\begin{itemize}}
\newcommand{\ei}{\end{itemize}}
\newcommand{\ben}{\begin{enumerate}}
\newcommand{\een}{\end{enumerate}}

\begin{document}

\title{The Secrecy Capacity of The Gaussian Wiretap Channel with Rate-Limited Help}

\author{Sergey Loyka, Neri Merhav

%\vspace*{-1\baselineskip}

%\thanks{This paper was presented in part at the 5th IEEE Global Conference on Signal and Information Processing, Montreal, Canada, Nov. 2017 \cite{Loyka-17-2}.}

\thanks{S. Loyka is with the School of Electrical Engineering and Computer Science, University of Ottawa, Ontario, Canada, e-mail: sergey.loyka@uottawa.ca.}
\thanks{N. Merhav is with the Andrew and Erna Viterbi Faculty of Electrical and Computer Engineering, Technion - Israel Institute of Technology, Haifa, Israel, e-mail: merhav@ee.technion.ac.il.}
\thanks{This paper was presented in part at the IEEE International Symposium on Information Theory, Helsinki, Finland, June 26 - July 1, 2022.}
}

%\vspace*{-2\baselineskip}

\maketitle

% to remove page numbers
%\pagenumbering{gobble}

%\vspace*{-1\baselineskip}
\begin{abstract}
The Gaussian wiretap channel with rate-limited help, available at the legitimate receiver (Rx) or/and transmitter (Tx), is studied under various channel configurations (degraded, reversely degraded and non-degraded). In the case of Rx help and all channel configurations, the rate-limited help results in a secrecy capacity boost equal to the help rate irrespective of whether the help is secure or not, so that the secrecy of help does not provide any capacity increase. The secrecy capacity is positive for the reversely-degraded channel (where the no-help secrecy capacity is zero) and no wiretap coding is needed to achieve it. More noise at the legitimate receiver can sometimes result in higher secrecy capacity. The secrecy capacity with Rx help is not increased even if the helper is aware of the message being transmitted. The same secrecy capacity boost also holds if non-secure help is available to the transmitter (encoder), in addition to or instead of the same Rx help, so that, in the case of the joint Tx/Rx help, one help link can be omitted without affecting the capacity. If Rx/Tx help links are independent of each other, then the boost in the secrecy capacity is the sum of help rates and no link can be omitted without a loss in the capacity. Non-singular correlation of the receiver and eavesdropper noises does not affect the secrecy capacity and non-causal help does not bring in any capacity increase over the causal one.
\end{abstract}

%\vspace*{-.8\baselineskip}

%=====================================================================================
\section{Introduction}

Physical-layer security has emerged as a valuable alternative to cryptography-based techniques \cite{Bloch-11}-\cite{Regalia-15}, especially over wireless channels and networks, and it also plays an important role in modern industrial standards \cite{Wu-18}-\cite{Chorti-22}. While the original work on information-theoretic secrecy dates back to Shannon himself \cite{Shannon-49}, Wyner's wiretap channel (WTC) model \cite{Wyner-75} established itself as a very useful tool for many different settings and configurations. It includes one legitimate transmitter-receiver pair and
one wiretapper (or eavesdropper) to be kept ignorant of the transmitted message; see \cite{Massey-83} for a simplified analysis of this model. Its key performance metric is the secrecy capacity, i.e. the largest achievable rate subject to (weak or strong) secrecy in addition to a reliability constraint, possibly under a power constraint. The original degraded WTC model has been extended and developed in many respects, of which we mention here only a few. Csiszar and Korner \cite{Csiszar-78} extended it to broadcast channels with confidential messages, including non-degraded WTC as a special case and established its secrecy capacity, which became a starting point for many further extensions and developments, see e.g. \cite{Bloch-11}-\cite{Wu-18} and references therein. The discrete memoryless model was extended to single-antenna (SISO) Gaussian settings in \cite{Leung-Yan-Cheong-78} and further to multi-antenna (MIMO) settings in \cite{Khisti-10-jul}-\cite{Oggier-11}; the respective secrecy capacities and optimal signalling strategies were also established \cite{Leung-Yan-Cheong-78}-\cite{Loyka-16b} and were further extended to interference-constrained channels \cite{Dong-18}, \cite{Dong-20}. The SISO Gaussian WTC with interference known to the transmitter was studied in \cite{Mitrpant-06} and its achievable secrecy rates were obtained. Encoding individual (deterministic) source sequences for the degraded memoryless WTC was studied in \cite{Merhav-21-Det}; a necessary condition for secure and reliable transmission of such sequences was obtained and an achievability scheme was also given. More refined performance metrics (beyond secrecy capacity), including secrecy exponents, finite blocklength and second-order coding rates, have also been studied, see \cite{Bastani-Parizi-17}, \cite{Yang-19} and references therein.

While the above models assume the availability of complete knowledge of the channel, such knowledge may be incomplete or inaccurate in many practical settings and a compound channel model emerges. Finite-state compound WTCs have been studied in \cite{Liang-09}, \cite{Bjelakovic-13} and their secrecy capacities were established under certain degradedness assumptions. The secrecy capacity of a class of compound Gaussian wiretap MIMO channels with normed uncertainty (not necessarily degraded or finite state) and an optimal signalling strategy were established in \cite{Schaefer-15}.

The original WTC model can also be extended in other respects, including the addition of side information and feedback, which are often available in modern systems and networks. While feedback does not increase the ordinary (no secrecy) capacity of memoryless channels, it is often able to boost the secrecy capacity, even in the memoryless settings, see e.g. \cite{Li-19} and references therein. The memoryless Gaussian WTC with noiseless (and hence rate-unlimited) feedback was considered in \cite{Gunduz-08}, whereby the transmitter (Tx) has access to the signal of the legitimate receiver (Rx) in a causal manner while the eavesdropper (Ev) has access to a noisy version of the feedback. Its secrecy capacity $C_{snf}$ was shown to be equal to the ordinary (no Ev, no feedback) AWGN channel capacity $C_0$,
\bal
\label{eq.Csnf}
C_{snf} = C_0
\eal
i.e. secrecy comes for free with the noiseless feedback and the secrecy capacity with feedback exceeds the no-feedback one, even though the channel is memoryless and, possibly, not degraded. The capacity-achieving strategy is the Schalkwijk-Kailath scheme \cite{Schalkwijk-66} (which is also optimal for the no-Ev/no-secrecy case) and no wiretap coding is needed. This result was further extended to a colored (ARMA) Gaussian noise channel with noiseless (rate-unlimited) feedback in \cite{Li-19} and a generalized Schalkwijk-Kailath scheme was shown to be optimal. Note, however, that, in this setting, the Tx has access to the noiseless feedback while the Ev observes only its noisy version, i.e. the Ev is at a significant disadvantage and the feedback is at least partially secure (hidden by the noise in the Ev feedback link). The situation changes dramatically if the Ev has access to the same noiseless (and, hence, non-secure) feedback as well or if the Rx-to-Tx feedback link is also noisy or rate-limited (the Schalkwijk-Kailath scheme does not work in this case).
The degraded memoryless Gaussian WTC with a secure \textit{rate-limited} feedback of rate $R_f< \infty$ was considered in \cite{Ardestanizadeh-09} and its secrecy capacity $C_{sf}$ was established:
\bal
\label{eq.Csf}
C_{sf} = \min\{C_0, C_{s0}+R_f\}
\eal
where $C_{s0}$ is the secrecy capacity without feedback\footnote{It follows that $C_{sf}=C_{snf}=C_0$ if the feedback rate is sufficiently large, $R_f \ge C_0-C_{s0}$, i.e. the increase in $C_{sf}$ with $R_f$ saturates at $C_{sf}=C_0$ and further increase in $R_f$ does not bring in any capacity increase so that the rate-unlimited feedback, as in \cite{Gunduz-08}, is not necessary to achieve $C_{sf}=C_0$.}. An optimal Tx strategy is fundamentally different from \cite{Li-19}, \cite{Gunduz-08} in this setting: it is a combination of the standard wiretap coding as in \cite{Wyner-75} with a secure  key generated by the Rx and sent to the Tx via the secure rate-limited feedback link (it is this 2nd part that is responsible for the $+R_f$ boost in the secrecy capacity as it protects a part of the message which was dummy in \cite{Wyner-75}). Note, however, that this strategy requires a secure feedback link, so that the feedback is (completely) unknown to the Ev, and it does not apply otherwise.

In modern communication systems/networks, various forms of side information, beyond feedback, are often available to the encoder or/and decoder (e.g. in a cloud radio access network with a centralized processing unit or in a cooperative communication system). This can be used to facilitate reliable communications and often results in a boost to the capacity \cite{Keshet-08}. One particular configuration was recently studied in \cite{Bross-20}-\cite{Marti-19}, where a rate-limited (and error-free) help is available to the decoder or/and encoder. In particular, a helper observes the noise sequence (which can be a signal intended for other users in a multi-user environment) and communicates his observation to the receiver (decoder) or transmitter (encoder) via an error-free but rate-limited data pipe. This model is, in our opinion, important from a practical perspective since it considers a rate-limited help, unlike some noiseless feedback models that essentially require rate-unlimited and error-free feedback links, which are hardly possible in practice. This rate-limited help was shown in \cite{Bross-20}-\cite{Marti-19} to provide a channel capacity boost equal to the help rate $R_h$ so that the resulting channel capacity is $C_0 +R_h$; flash signalling (i.e. using high-resolution help infrequently) was shown to be an optimal help strategy, in combination with two-phase time sharing.  Error exponents of Gaussian and modulo-additive channels with rate-limited Tx help were established in \cite{Merhav-21}, where it was also shown that the channel with Tx help is equivalent, in this respect, to the regular (no-help) channel and an additional parallel error-free bit-pipe of rate $R_h$.

In the present paper, we extend the help setting in \cite{Bross-20}-\cite{Marti-19} to the memoryless Gaussian wiretap channel. In the case of Rx help, we show that the same capacity boost as in \cite{Bross-20} also holds for the wiretap channel in terms of its secrecy capacity $C_s$: a receiver help of rate $R_h$ results in the secrecy capacity boost of $R_h$, i.e. the $+R_h$ capacity boost come with secrecy for free,
\bal
\label{eq.Cs.intro}
C_s = C_{s0} + R_h,
\eal
where $C_{s0}$ is the no-help secrecy capacity (if noises are not Gaussian, then the rate boost is upper bounded by $R_h$ and $C_s \le C_{s0} + R_h$). This holds for all possible configurations of the SISO Gaussian WTC, i.e. degraded, reversely degraded and non-degraded\footnote{While the standard (no help) SISO non-degraded Gaussian WTC is equivalent to either degraded or reversely-degraded one, this is not the case anymore when Rx/Tx help is also available to the Ev.}, with only one exception. Some surprising properties are observed. In particular, the secrecy capacity is the same irrespective  of whether the help is secure (i.e. unknown to the eavesdropper) or not, so that the secrecy of help does not bring in any increase in the secrecy capacity; this also applies to the case of partially-secure help. For the reversely-degraded channel (where the secrecy capacity is zero without help), we show that the secrecy capacity with Rx help is positive and equal to the help rate, that no  wiretap coding is needed to achieve it, and that burst signaling (along with regular channel coding) is optimal. Unlike the no-help case, more noise at the legitimate receiver can sometimes result in higher secrecy capacity. Surprisingly, the secrecy capacity with Rx help, secure or non-secure, is not increased even if the helper is aware of the message being transmitted.

An optimal Tx strategy to achieve $C_s$ in \eqref{eq.Cs.intro} is fundamentally different from those in \cite{Li-19}, \cite{Gunduz-08}, \cite{Ardestanizadeh-09}: it is a two-phase time sharing whereby no help is used in Phase 1 but just regular (no help) wiretap coding; on the contrary, much shorter Phase 2 makes use of high-resolution help in combination with regular (no Ev) channel coding but no wiretap coding at all. In the case of the reversely-degraded WTC, Phase 1 and hence wiretap coding are not needed so that burst signaling alone (with regular channel coding) is sufficient.

Comparing \eqref{eq.Cs.intro} to \eqref{eq.Csf} with $R_h=R_f$, note that $C_s > C_{sf}$ if the help/feedback rate is sufficiently high, $R_h=R_f > C_0-C_{s0}=C_2$, where $C_2$ is the no-help capacity of the Tx-Ev link, i.e. the helper setting provides larger secrecy capacity compared to the rate-limited but secure feedback setting, even though the help is not required to be secure. The same applies to \eqref{eq.Csnf}, where the feedback is rate-unlimited and at least partially-secure. Note also that, unlike $C_{sf}$ in \eqref{eq.Csf}, the increase in $C_s$ in \eqref{eq.Cs.intro} with $R_h$ does not saturate.

We further show that, in the case of the degraded or reversely-degraded Gaussian WTC, the same secrecy capacity boost, and hence \eqref{eq.Cs.intro}, hold when non-secure help is available to the transmitter, in addition to or instead of the same Rx help, and an optimal signalling is still two-phase time sharing. Thus, if the Tx and Rx help links are identical (carry the same information), then any one can be omitted without affecting the capacity. This is not the case anymore if the help links are independent: in this case, the secrecy capacity boost is the sum of help rates, an optimal signalling is a three-phase time sharing and no help link can be omitted without capacity loss.

While causality is immaterial for Rx help (since the receiver starts decoding after the whole block of symbols is received), it becomes important for the Tx help since the transmitter performs sequential symbol-by-symbol transmission. Therefore, we distinguish between causal and  non-causal Tx help. In the latter case, the help is based on the whole noise sequence and is available to the Tx in advance. In the former case, the Tx help at time $i$ is based on the noise sequence up to time $i$ only. Interestingly, the causality of Tx help, unlike that of feedback, has no impact on the secrecy capacity (this property is similar to that of the no-secrecy channel capacity with Tx help in \cite{Lapidoth-20}).

Unlike the studies of Gaussian WTCs with noiseless (and hence rate-unlimited) feedback in \cite{Li-19}, \cite{Gunduz-08}, our help links are rate-limited, as in \cite{Bross-20}-\cite{Marti-19}, and we also allow here the Ev to have access to the same help as the legitimate Rx or/and Tx (in the case of non-secure help). In our rate-limited setting, causality of help has no impact on the secrecy capacity and, in the case of Rx help, the secrecy capacity is the same for perfectly secure and completely non-secure help (i.e. when exactly the same help is also available to the Ev). Unlike the study in \cite{Ardestanizadeh-09}, our help link is not required to be secure or causal and the channel is not required to be degraded.

In a related line of work, secure communication with a helper acting as a cooperating jammer was studied in \cite{Fritschek-16}, \cite{Chen-20} (this setting is partialy equivalent to an interference channel). However, no secrecy capacity was established but only the generalized degrees of freedom (GDoF), which characterize the high-SNR scaling of the secrecy capacity and are essentially the multiplexing gain in terms of secrecy rates. Unlike \cite{Fritschek-16}, \cite{Chen-20}, the present paper considers no jamming at all; rather, the help comes in a form of rate-limited error-free information about the noise sequence affecting the legitimate Rx, which is available to the Rx and/or Tx.

The rest of the paper is organized as follows. Various configurations (degraded, reversely-degraded and non-degraded) of the Gaussian WTC with Rx help are considered in Sections \ref{sec.D-WTC} to \ref{sec.Non-Deg.Ch} and their secrecy capacities are established in Theorems \ref{thm.Cs} - \ref{thm.Cs.Non-Deg} and Propositions 1, 2. The case of Tx help, instead of or in addition to the Rx help, is studied in Sections \ref{sec.D.Tx} - \ref{sec.Ind.TxRxH} and the respective secrecy capacities are established/characterized in Theorems \ref{thm.TxH.Cs} - \ref{thm.TxRxH-RD-I.Cs} and Proposition 3, including the same and independent Tx/Rx help links, and the case of correlated Rx and Ev noises.

\textit{Notations}: we follow the standard notations as much as possible, where random variables and their realizations are denoted by capital and lower case letters, respectively, and their alphabets follow from the respective channel models; $X^n$ denotes the sequence $(X_1,...,X_n)$; $H(\cdot)$, $h(\cdot)$ and $h(\cdot|\cdot)$ are the entropy, differential and conditional differential entropies, respectively, and $I(\cdot;\cdot)$ is the mutual information; $\mathbb{E}\{\cdot\}$ and $\Pr\{\cdot\}$ are statistical expectation and probability with respect to relevant random variables; $X-Y-Z$ denotes a Markov chain of random variables $X$, $Y$, and $Z$.

%\newpage

%\vspace*{-.51\baselineskip}
%======================================================================================
\section{Degraded Gaussian Wiretap Channel With Rx Help}
\label{sec.D-WTC}

We begin with the real-valued degraded (discrete-time) Gaussian wiretap channel:
\bal
Y_i = X_i + W_i,\ Z_i = Y_i + V_i,\ i=1,...,n
\eal
where $X_i$ is the real-valued transmitted symbol at time $i$, $W_i,\ V_i$ are Rx and Ev noise, which are zero-mean Gaussian, independent of each other with variances $\sigma^2_W$ and $\sigma^2_V$, respectively, see Fig. \ref{fig.RxH}. The channel is stationary and memoryless, so that $W^n$ and $V^n$ are i.i.d. sequences. We further assume that $\sigma_V^2 >0$ (unless stated otherwise). A slightly more general case, where noises are not Gaussian, will also be considered.

The helper model is as in \cite{Bross-20}-\cite{Marti-19} but extended to the WTC setting, whereby discrete help $T=T(W^n)$ of rate $n^{-1} H(T)\le R_h < \infty$ is available to the Rx and Ev (no further constraints on the helper function $T(W^n)$ are assumed, beyond its rate, unless stated otherwise), which we term "non-secure Rx help", so that the Rx and the Ev can estimate transmitted message $M$ using $T$ and their respective received signals $Y^n$ and $Z^n$. This falls into the framework of cooperative communications or communications with side information \cite{Keshet-08} and models practical links, which are always rate-limited (albeit the rate can be high, as in e.g. optical fiber links).
If no help is available to the Ev, we call it "secure Rx help". For Rx help, the difference between causal and non-causal help is immaterial, since the Rx waits until the whole block of length $n$ is received before decoding it.

We use the standard definition of the secrecy capacity as the supremum of all achievable secrecy rates, subject to the reliability, secrecy and power constraints, see e.g. \cite{Bloch-11}-\cite{Wu-18}, \cite{Wyner-75}-\cite{Csiszar-78}.
In particular, the (secret) message $M$ is selected randomly and uniformly from $\{1,...,2^{nR_s}\}$, where $R_s$ is a secrecy rate and $n$ is the blocklength. The Tx encoder maps it into $X^n$ and the Rx decoder maps $Y^n$ and the available help $T$ into a message estimate $\hat{M}$. The constraints are as follows:

\textit{Reliability constraint}: the error probability $P_e \triangleq \Pr\{M \neq \hat{M}\} \le \varepsilon$ for any $\varepsilon >0$ and sufficiently large $n$.

\textit{Weak secrecy constraint}: information leakage rate (to the Ev) $R_l$ satisfies
\bal
R_l \triangleq n^{-1} I(M;Z^nT) \le \delta
\eal
for any $\delta >0$ and sufficiently large $n$; $T$ is omitted in the case of secure help.

\textit{Average power constraint}:
\bal
\label{eq.TPC}
\frac{1}{n} \sum_{i=1}^n \mathbb{E}\{X_i^2\} \le P
\eal
We further assume that $P>0$ (otherwise, the capacity is, of course, zero).
\begin{figure}[t]
    \begin{center}
        \includegraphics[width=3.5in]{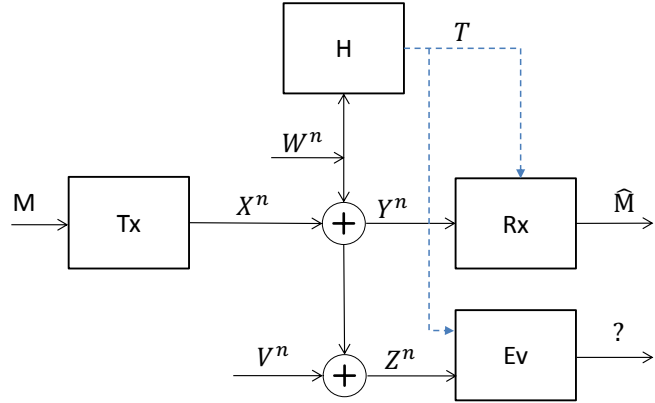}
    \end{center}
    \caption{Degraded wiretap channel with a rate-limited help $T$ at the Rx and Ev (if the help is not secure). $W^n$ and $V^n$ are i.i.d. noise sequences independent of each other and of $M, X^n$; $X^n=X^n(M)$, $T=T(W^n)$, $H(T)\le n R_h$.}
    \label{fig.RxH}
\end{figure}

The secrecy capacity of this channel with Rx help is established below.

%\newpage
%=====================================================================================
%\subsection{Secrecy Capacity with Rx Help}
%\label{sec.C.RxH}

\begin{thm}
\label{thm.Cs}
The secrecy capacity $C_{s}$ of the degraded memoryless WTC (not necessarily Gaussian) with secure or non-secure Rx help of rate $R_h$, as in Fig. \ref{fig.RxH}, is bounded as follows
\bal
C_{s} \le C_{s0} + R_h
\eal
where $C_{s0}$ is the secrecy capacity without help. This holds with equality if the noises are Gaussian and $\sigma_V^2 >0$, for which $C_{s0}=C_1-C_2$, where $C_1$, $C_2$ are the capacities of the Tx-Rx and Tx-Ev links without help.
\end{thm}
\begin{proof}

\textit{Converse}: For the converse, we do not assume that the noises are Gaussian and consider the case of secure Rx help $T$ (not available to the Ev); the case of non-secure help will follow since the availability of help to the Ev cannot increase secrecy rate. The converse is based on the following chain of inequalities, incorporating the secrecy and reliability constraints as well as functional relationships between various random variables in the channel model:
\bal\notag
n R_s &= H(M)\\
&= H(M|Z^n) + I(M;Z^n)\\
\label{eq.RxH.p1}
&\le H(M|Z^n) + n\epsilon\\
%\label{eq.TxH.p2}
&= H(M|Z^n)- H(M|Y^n T) + H(M|Y^n T) + n\epsilon\\
\label{eq.RxH.p2}
&\le H(M|Z^n)- H(M|Y^n T)  + 2n\epsilon\\
&\le H(M|Z^n)- H(M|Y^n Z^n T)  + 2n\epsilon\\
\label{eq.RxH.p4}
&= I(M;Y^n T|Z^n) + 2n\epsilon\\
\label{eq.RxH.p5}
&\le I(X^n;Y^n T|Z^n) + 2n\epsilon\\
%\label{eq.TxH.p6}
&= I(X^n;Y^n |Z^n) + I(X^n;T|Y^n Z^n) +2n\epsilon\\
\label{eq.RxH.p7}
&\le I(X^n;Y^n |Z^n) + H(T) +2n\epsilon\\
\label{eq.RxH.p8}
&\le n I_0(X;Y|Z) + nR_h +2n\epsilon\\
\label{eq.RxH.p9}
&= n (I_0(X;Y) - I_0(X;Z)+ R_h +2\epsilon)\\
\label{eq.RxH.p10}
&\le n (C_{s0} + R_h +2\epsilon)
\eal
where \eqref{eq.RxH.p1} follows from the secrecy constraint $I(M;Z^n) \le n\epsilon$; \eqref{eq.RxH.p2} follows from Fano inequality (due to the reliability constraint) $H(M|Y^n T) \le n \epsilon$;
\eqref{eq.RxH.p5} follows from Markov chain $M - X^n - Y^nT - Z^n$;
\eqref{eq.RxH.p7} follows from $I(X^n;T|Y^n Z^n) \le H(T)$;
\eqref{eq.RxH.p8} follows from
\bal
I(X^n;Y^n |Z^n) \le \sum_{i=1}^n I(X_i;Y_i|Z_i) &\le n I_0(X;Y|Z)
\eal
where the first inequality holds since the channel is memoryless and the second one is due to the concavity of the mutual information in the input distribution \cite{Wyner-75}\cite{Massey-83}; $I_0$ is the mutual information induced by input $X$ with the distribution $p_0(x) = n^{-1}\sum_i p_{x_i}(x)$, where $p_{x_i}(x)$ is the distribution of $X_i$; \eqref{eq.RxH.p9} follows from Markov chain $X - Y - Z$.
Thus,
\bal
R_s \le C_{s0} + R_h +2\epsilon
\eal
Since this holds for any $\epsilon >0$, it follows that $R_s \le C_{s0} + R_h$. This establishes the converse with secure Rx help. Since the presence of help at Ev cannot increase secrecy rate, the same upper bound applies with non-secure Rx help.

\textit{Achievability}. To prove achievability, we assume that the noises are Gaussian and combine the regular (no help) wiretap coding with the no-Ev flash signaling in \cite{Bross-20}. We consider first the case of non-secure Rx help (when the same help is available at the Rx and Ev), from which achievability with secure Rx help follows. To this end, recall that the ordinary (no Ev) flash signaling with Rx help consists of two phases of time-sharing \cite{Bross-20}:

\bi
\item Phase 1: no help is used at all for a fraction $(1-\tau)$ of the time, which achieves, with regular channel coding, a rate arbitrary close to the ordinary channel capacity $C$ for a sufficiently large blocklength.
\item Phase 2: Rx help is used at rate $R_h/\tau$ for a (very small) fraction $\tau$ of the time. In this phase, in addition to regular channel coding, a high-resolution scalar quantization (with $\lfloor2^{R_h/\tau}\rfloor$ levels) of each noise sample is provided to the Rx, so that the help is $T = \hat{W}^n$, where $\hat{W}_i = Q(W_i)$ and  $Q(\cdot)$ is a scalar quantizer. The Rx subtracts $\hat{W}_i$ from its received signal $Y_i$ and, after receiving the whole block, decodes it using nearest-neighbour decoding; for sufficiently large blocklength, this achieves a rate arbitrarily close to
\bal
\label{eq.RxH.Ph2.rate}
h(X) - \frac{1}{2}\log(2\pi e \sigma_W^2) + \frac{R_h}{\tau} = \frac{R_h}{\tau}(1+o(1))
\eal
where $o(1) \to 0$ as $\tau \to 0$, see \cite[eq. (20)]{Bross-20}. An alternative Phase 2 strategy, which maximizes error exponents using a simple lattice code with a uniform scalar quantizer (no need for i.i.d.-generated codebooks), can be found in \cite{Merhav-21}.
\ei
Overall, as $\tau \to 0$, the rate achieved after two-phase time-sharing is arbitrarily close to
\bal
(1-\tau)C + \tau R_h/\tau(1+o(1)) \to C +R_h
\eal
which is the channel capacity with Rx help. This also implies that providing high-resolution help infrequently ("flash signalling") is optimal.

To accommodate the Ev and the secrecy constraint, we modify this strategy as follows:
\bi
\item Phase 1: use the regular WTC coding with no help \cite{Bloch-11}-\cite{Wu-18}\cite{Wyner-75}\cite{Massey-83}\cite{Leung-Yan-Cheong-78} for the fraction $(1-\tau)$  of the time; this achieves a secrecy rate $R_s$ arbitrarily close to the regular WTC secrecy capacity $C_{s0}$: $R_s = C_{s0} - \epsilon$ for any $\epsilon >0$ and sufficiently large blocklength.
\item Phase 2: for the fraction $\tau$ of the time, use no WTC coding but ordinary channel coding under the flash signaling as above.
\ei
While it is clear that secrecy is guaranteed during Phase 1 (via wiretap coding), it is also clear that secrecy is not guaranteed during Phase 2 (since no wiretap coding is used) so it is not clear whether  secrecy is guaranteed overall (after time sharing). To demonstrate that this is indeed the case, we show that, during Phase 2, the information leakage rate $R_{l2}$ to the Ev is uniformly bounded,
\bal
R_{l2} \le R_0 <  \infty
\eal
for any $\tau$ and some $R_0$, where $R_0$ is independent of $\tau$ (but where $R_{l2}$ may depend on $\tau$), so that the overall leakage rate $R_l$ (after the time sharing) is
\bal
\label{eq.Rl}
R_l= (1-\tau)R_{l1} + \tau R_{l2} \le (1-\tau)\delta + \tau R_{0} \to \delta
\eal
as $\tau \to 0$, for any $\delta >0$ (or, equivalently, $R_l \le 2\delta$ for sufficiently small $\tau$, $\tau \le \delta/R_0$), where $R_{l1} \le \delta$ is the information leakage rate during Phase 1.

To see that indeed $R_{l2} \le R_0 < \infty$ uniformly in $\tau$, note the following:
\bal
R_{l2} &= n^{-1}I(M_2;Z^n \hat{W}^n|\mathcal{C})\\
&\le n^{-1}I(M_2;Z^n \hat{W}^n {W}^n|\mathcal{C})\\
\label{eq.AP.3}
&= n^{-1}I(M_2;Z^n|{W}^n\mathcal{C})\\
\label{eq.AP.4}
&\le n^{-1}I(X^n;Z^n|W^n\mathcal{C})\\
&= n^{-1}I(X^n;X^n +W^n +V^n|W^n\mathcal{C})\\
\label{eq.AP.7}
&= n^{-1}I(X^n;X^n+V^n|\mathcal{C})\\
&\le \frac{1}{2}\log\left(1+\frac{P}{\sigma_V^2}\right) = C_2' < \infty
\eal
where $M_2$ is a message sent in Phase 2, $X^n$ is a codeword (which depends on $M_2$, see Fig. \ref{fig.RxH}), and the conditioning is on an i.i.d. randomly-generated codebook $\mathcal{C}$ (the codebook generation, encoding and decoding are as in \cite{Bross-20}); \eqref{eq.AP.3} follows from independence of $M_2$ and $W^n, \hat{W}^n$ and from $\hat{W}_i = Q(W_i)$; \eqref{eq.AP.4} follows from the Markov chain $M_2 - X^n - Z^nW^n$; \eqref{eq.AP.7} follows from independence of $W^n$ and $X^n, V^n$.

Hence, arbitrary low information leakage rate is guaranteed after time sharing with $\tau \to 0$, which satisfies the secrecy constraint. At the same time, the overall secrecy rate (after time sharing) is
\bal\notag
\label{eq.Rs.RxH.Deg}
(1-\tau)(C_{s0} - \epsilon) &+ \tau R_h/\tau(1+o(1))\\
  &\to C_{s0} +R_h -\epsilon
\eal
for any $\epsilon >0$, as $\tau \to 0$, so that the secrecy capacity is $C_{s0} +R_h$, as required.

In the above secrecy analysis, we assume that the help is not secure, i.e. it is available to the Ev. Clearly, the secrecy constraint is also satisfied if the help is secure, i.e. not available to the Ev (since the lack of Ev help cannot increase leakage rate), and an achievable  secrecy rate remains the same. Since the converse also holds for the secure Rx help, the secrecy capacity also remains the same, regardless whether help is secure or not, i.e. the secrecy of help does not increase the secrecy capacity.
\end{proof}

It is worthwhile to note that flash signaling with Rx help provides here the same boost in the secrecy capacity as in the regular (no Ev)  channel capacity in \cite{Bross-20}, i.e. the $+ R_h$ boost comes with secrecy for free in the degraded Gaussian WTC. If the noises are not Gaussian, then the rate boost is upper bounded by $R_h$ (it remains to be seen whether it is actually equal to $R_h$).

Since $C_s$ in Theorem 1 is the same for secure and non-secure help, i.e. the secrecy of help does not bring in any capacity advantage, it also applies to the case of partially-secure help, i.e. when the Ev has access to a part of $T$.

Surprisingly, even if the helper \textsf{H} is aware of the message $M$ being transmitted, i.e. $T=T(W^n,M)$ as in Fig. \ref{fig.RxHEvHM}, the secrecy capacity is not affected and Theorem \ref{thm.Cs} still holds.

\begin{figure}[t]
    \begin{center}
        \includegraphics[width=3.5in]{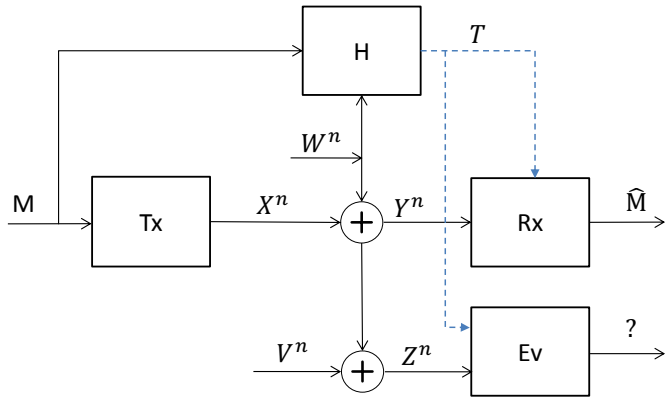}
    \end{center}
    \caption{The degraded WTC of Fig. \ref{fig.RxH} when the helper \textsf{H} is aware of the message $M$ being sent, $T=T(W^n,M)$.}
    \label{fig.RxHEvHM}
\end{figure}

\begin{prop}
Consider the degraded WTC with Rx help as in Theorem \ref{thm.Cs} and let the helper \textsf{H} be aware of the message being transmitted, i.e. $T=T(W^n,M)$ as in Fig. \ref{fig.RxHEvHM}. Then, Theorem \ref{thm.Cs} still holds.
\end{prop}
\begin{proof}
It is sufficient to show that the same converse still holds (for achievability, the helper can always ignore the message). To this end, note that \eqref{eq.RxH.p1} - \eqref{eq.RxH.p4} still hold since the independence of $T$ and $M$ plays no role there so that
\bal\notag
n R_s &\le I(M;Y^n T|Z^n) + 2n\epsilon\\
%\label{eq.RxHM.p2}
&\le I(X^n M;Y^n T|Z^n) + 2n\epsilon\\
%\label{eq.TxH.p6}
&= I(X^n M;Y^n |Z^n) + I(X^n M;T|Y^n Z^n) +2n\epsilon\\
%\label{eq.RxH.p7}
&\le I(X^n M;Y^n |Z^n) + H(T) +2n\epsilon\\
\label{eq.RxHM.p5}
&= I(X^n;Y^n |Z^n) + H(T) +2n\epsilon\\
\label{eq.RxHM.p6}
&\le n (C_{s0} + R_h +2\epsilon)
\eal
where \eqref{eq.RxHM.p5} is due to $I(M;Y^n|X^n Z^n)=0$, i.e., the independence of $M$ and $Y^n$ given $X^n$ and $Z^n$; \eqref{eq.RxHM.p6} follows from \eqref{eq.RxH.p8}-\eqref{eq.RxH.p10}.
\end{proof}

For the degraded Gaussian WTC with non-secure Rx help, the secrecy capacity is zero if $\sigma_V^2=0$, since the Ev has access to the same information as the Rx in this case, so no secrecy is possible. This implies that $C_s(\sigma_V^2)$ is a discontinuous function at $\sigma_V^2 =0$ for non-secure  help with any $R_h >0$:
\bal
\lim_{\sigma_V^2 \to 0^+} C_s(\sigma_V^2) = R_h >0,
\eal
while $C_s(0)=0$ (the help becomes useless in this case). This is in stark contrast to the no-help case where $C_{s0}(\sigma_V^2)$ is a continuous function for every $\sigma_V^2$, including $\sigma_V^2=0$, so that the help becomes especially important when $\sigma_V^2$ approaches 0, i.e., when the Tx-Ev link SNR  approaches that of the Tx-Rx link.

%\vspace*{-.51\baselineskip}
%======================================================================================
\section{Reversely-Degraded WTC With Rx Help}
\label{sec.Rev.Deg.Ch}

Let us now consider the reversely-degraded case  of the wiretap channel as in Fig. \ref{fig.RxH-RD}:
\bal
Z_i = X_i + V_i,\ Y_i = Z_i + \Delta W_i
\eal
where $\Delta W_i$ is an extra Rx noise, independent of the Ev noise $V_i$, so that the sequences $V^n$ and $\Delta W^n$ are i.i.d and independent of each other. Note that the total Rx noise is ${W}_i = V_i + \Delta W_i$ and its variance is
\bal
\sigma_{{W}}^2 = \sigma_{V}^2 + \sigma_{\Delta W}^2 \ge \sigma_{V}^2 > 0
\eal
We exclude the trivial case $\sigma_{V}^2 = 0$, for which the secrecy capacity is zero. It is well-known that, without help, the secrecy capacity of this channel is zero, $C_{s0}=0$. However, the availability of Rx help, whether or not securely, changes the situation dramatically.

\begin{figure}[t]
    \begin{center}
        \includegraphics[width=3.5in]{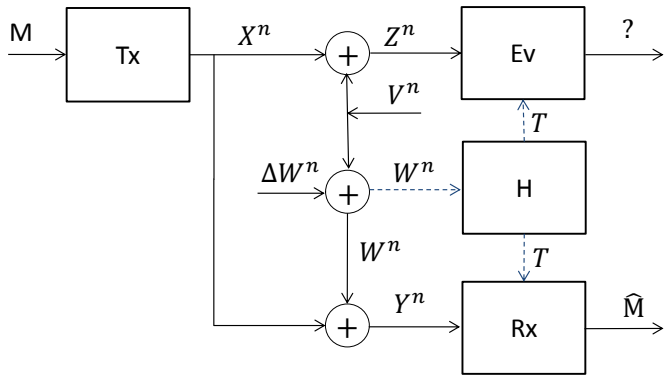}
    \end{center}
    \caption{Reversely-degraded wiretap channel with a rate-limited Rx help $T$. $\Delta W^n$ and $V^n$ are i.i.d. noise sequences independent of each other and of $M, X^n$; $X^n=X^n(M)$, $T=T(W^n)$, $H(T)\le n R_h$.}
    \label{fig.RxH-RD}
\end{figure}

\begin{thm}
\label{thm.Cs.RevDeg}
The secrecy capacity $C_{s}$ of the reversely-degraded WTC in Fig. \ref{fig.RxH-RD} (not necessarily Gaussian) with secure or non-secure Rx help of rate $R_h$ is bounded as follows
\bal
C_{s} \le  R_h
\eal
This holds with equality if the noises are Gaussian and $\sigma_V^2, \sigma_{\Delta W}^2 > 0$; if $\sigma_{\Delta W}^2 = 0$, then $C_{s} =  R_h$ if the help is secure and $C_{s} =  0$ otherwise.
\end{thm}
\begin{proof}
\textit{Converse:} to prove the converse, we do not assume that the noises are Gaussian and consider the case of secure Rx help (i.e. no Ev help). The case of non-secure help will follow, since the availability of help to the Ev cannot increase the secrecy rate. The proof follows the steps similar to those in Theorem \ref{thm.Cs}. In particular, we observe that \eqref{eq.RxH.p1}-\eqref{eq.RxH.p7} still hold for the reversely-degraded channel (since channel degradedness plays no role there), so that
\bal
\label{eq.RxH.RD.p1}
n R_s &\le I(X^n;Y^n|Z^n) + H(T) + 2n\epsilon\\
\label{eq.RxH.RD.p2}
&\le n(R_h + 2\epsilon),
\eal
where the last inequality is due to $I(X^n;Y^n|Z^n)=0$, which in turn follows from Markov chain $X^n - Z^n - Y^n$. Thus, $R_s \le R_h + \epsilon$ for any $\epsilon >0$ and therefore $R_s \le R_h$, as required.

\textit{Achievability:} to prove achievability, we assume that the noises are Gaussian and consider the case of non-secure Rx help (when the same help is also available to the Ev); the case of secure help will follow since the absence of help to the Ev cannot increase leakage rate and hence cannot decrease secrecy rate.

To this end, we use the same two-phase flash signaling as in Theorem \ref{thm.Cs} except that nothing is transmitted in Phase 1 and the whole message is transmitted in Phase 2 (without wiretap coding). To show that this provides an arbitrary-low leakage rate after time-sharing (which is equivalent to burst signaling of duration $\tau$ in this case), we show that the Phase 2 leakage rate $R_{l2}$ is uniformly bounded in $\tau$ (as before). To this end, we observe that
\bal
\label{eq.APRD.1}
R_{l2} &= n^{-1}I(M_2;Z^n \hat{W}^n|\mathcal{C})\\
&\le n^{-1}I(M_2;Z^n \hat{W}^n {{W}}^n|\mathcal{C})\\
\label{eq.APRD.3}
&= n^{-1}I(M_2;Z^n|{{W}}^n\mathcal{C})\\
\label{eq.APRD.4}
&\le n^{-1}I(X^n;Z^n|{W}^n\mathcal{C})\\
&= n^{-1}I(X^n;X^n +V^n|V^n + \Delta W^n,\mathcal{C})\\
\label{eq.APRD.6}
&= n^{-1}(h(X^n+V^n|V^n + \Delta W^n,\mathcal{C})\\ \notag
 &\qquad - h(V^n|X^n, V^n + \Delta W^n,\mathcal{C}))\\
\label{eq.APRD.7}
&\le n^{-1}(h(X^n+V^n) - h(V^n|V^n + \Delta W^n)\\
\label{eq.APRD.8}
&\le \frac{1}{2}\log\left(1+\frac{P}{\sigma_V^2}\right) +  \frac{1}{2}\log\left(1+\frac{\sigma_V^2}{\sigma_{\Delta W}^2}\right) < \infty
\eal
where we assumed that $\sigma_{\Delta W}^2 > 0$; \eqref{eq.APRD.1}-\eqref{eq.APRD.4} hold due to the same reasons as in the proof of Theorem \ref{thm.Cs}; \eqref{eq.APRD.7} holds since (i) conditioning cannot increase the entropy and (ii) $V^n, \Delta W^n$ are independent of $X^n, \mathcal{C}$; \eqref{eq.APRD.8} holds since (i) the entropy is maximized by Gaussian distribution and
\bal\notag
h(V^n&|V^n + \Delta W^n)\\
 &= h(V^n,V^n + \Delta W^n) - h(V^n + \Delta W^n)\\
\label{eq.APRD.10}
&= h(V^n) +h(\Delta W^n) - h(V^n + \Delta W^n)\\
\label{eq.APRD.11}
&= \frac{n}{2}\log\frac{\sigma_V^2}{\sigma_V^2+ \sigma_{\Delta W}^2} +  \frac{n}{2}\log(2\pi e \sigma_{\Delta W}^2)
\eal
where \eqref{eq.APRD.10} is due to the independence of $\Delta W^n$ and $V^n$. Thus, the total leakage rate (after time-sharing) is
\bal
R_{l}&= (1-\tau)0 + \tau R_{l2}\\
&\le  \frac{\tau}{2}\log\left(1+\frac{P}{\sigma_V^2}\right) +  \frac{\tau}{2}\log\left(1+\frac{\sigma_V^2}{\sigma_{\Delta W}^2}\right) \to 0
\eal
when $\tau \to 0$, as required (notice that the condition $\sigma_{\Delta W}^2 >0$ is essential here, as $\sigma_{\Delta W}^2 =0$  results in zero secrecy capacity for non-secure help). The overall secrecy rate (after time-sharing) is
\bal
\label{eq.Rs.RD}
R_{s}&= (1-\tau)0 + \tau R_{h}/\tau (1+o(1)) \to R_h
\eal
when $\tau \to 0$.

Let us now consider the case of  $\sigma_{\Delta W}^2 = 0$, which implies $Y^n=Z^n$. If the help is not secure, the same information is available to the Ev and Rx and hence no positive secrecy rate is achievable, $C_s=0$. However, if the help is secure, then the Rx has an extra information not available to the Ev. It is not difficult to see that the above converse still holds if $\sigma_{\Delta W}^2 = 0$. To prove achievability, we use the same signaling as above and show that the leakage rate $R_{l2}$ of Phase 2 is uniformly bounded:
\bal
\label{eq.APRD0.1}
R_{l2} &= n^{-1}I(M_2;Z^n|\mathcal{C})\\
%\label{eq.APRD0.2}
&\le n^{-1}I(X^n;Z^n|\mathcal{C})\\
\label{eq.APRD0.3}
&\le n^{-1}(h(Z^n) - h(V^n))\\
%\label{eq.APRD.8}
&\le \frac{1}{2}\log\left(1+\frac{P}{\sigma_V^2}\right) = C_2< \infty
\eal
Thus, secrecy is guaranteed after time-sharing with $\tau \to 0$ and the achieved secrecy rate is as in \eqref{eq.Rs.RD}.
\end{proof}

It may feel counter-intuitive that $C_s =R_h >0$ for the reversely-degraded Gaussian WTC, even if the help is not secure, i.e. also available to the Ev, since, in this case, the Ev is getting more information than the Rx. However, one should also note that, even though the Ev has the right (public) "key" $T=\hat{W}^n$, it does not have the right "lock" $W^n$ to which this key applies and hence it cannot "unlock" it (i.e., cancel its own noise), unlike the legitimate Rx.

A related surprising observation follows from Theorem \ref{thm.Cs.RevDeg}: in the case of non-secure help, $C_s=0$ if $\sigma_W^2 = \sigma_V^2$ (i.e. $\sigma_{\Delta W}^2 =0$) but $C_s =R_h >0$ if $\sigma_W^2 > \sigma_V^2$, so that more noise at the legitimate Rx is actually better for secrecy in this case. This is due to the fact that the extra Rx noise $\Delta W_i \neq 0$ makes it impossible for the Ev to cancel its own noise using non-secure help $\hat{W}^n$ in the same way the Rx does (since $V_i \neq W_i$ in this case). However, if $\Delta W_i = 0$, then the Ev can do noise cancellation in the same way the Rx does, which results in $C_s=0$ and renders the help useless. This also implies that $C_s(\sigma_W^2)$ is a discontinuous function at $\sigma_W^2=\sigma_V^2$.

To summarize, the secrecy capacity $C_s$ of the degraded or reversely degraded Gaussian wiretap channel with Rx help of rate $R_h$ (secure or not) is given by
\bal
C_s = C_{s0} +R_h
\eal
if either $\sigma_W^2 \neq \sigma_V^2$ or else the help is secure, where, of course, $C_{s0}=0$ for the reversely-degraded case. Thus, not only the secrecy capacity is boosted by $R_h$ for the degraded case, but also the secrecy capacity is positive for the reversely-degraded case, where it is zero without help, and this positive secrecy capacity is achievable by burst signalling without any wiretap coding at all.

Similarly to the degraded WTC, Theorem \ref{thm.Cs.RevDeg} still holds even if  the helper \textsf{H} is aware of the message $M$ being transmitted, $T=T(W^n,M)$, so that there is no boost in the secrecy capacity due to the message being available to the helper.

\begin{prop}
Consider the reversely-degraded WTC with Rx help as in Theorem \ref{thm.Cs.RevDeg} and let the helper \textsf{H} be aware of the message being transmitted, i.e. $T=T(W^n,M)$. Then, Theorem \ref{thm.Cs.RevDeg} still holds.
\end{prop}
\begin{proof}
The converse follows since \eqref{eq.RxH.RD.p1}, \eqref{eq.RxH.RD.p2} still hold for $T=T(W^n,M)$. The achievability holds since the helper can always ignore the message.
\end{proof}

%\vspace*{-.51\baselineskip}
%======================================================================================
\section{Non-Degraded WTC With Rx Help}
\label{sec.Non-Deg.Ch}

Let us now consider the case where the channel is neither degraded nor reversely-degraded, as in Fig. \ref{fig.RxH-ND}:
\bal
Z_i = X_i + V_i,\ Y_i = X_i + W_i
\eal
where the noise sequences $V^n$ and $W^n$ are i.i.d Gaussian but possibly correlated with each other and the covariance matrix of $(W_i,V_i)$ is
\bal\notag
\label{eq.RWV}
\bR_{WV} &= \mathbb{E} (W_i, V_i) (W_i, V_i)'\\
 &= \left[
      \begin{array}{cc}
        \sigma_W^2 & r\sigma_W \sigma_V \\
        r\sigma_W \sigma_V & \sigma_V^2 \\
      \end{array}
    \right]
\eal
where $r$ is the normalized correlation coefficient, $|r| \le 1$, and $(\cdot)'$ means transposition. This correlation may be due to e.g. an external user's signal acting as the noise affecting the Rx and Ev. We further assume that its covariance matrix is not singular, i.e. the determinant $|\bR_{WV}| \neq 0$, which is equivalent to $|r| <1$. If $r=0$, then $V^n$ and $W^n$ are independent of each other.

It is well-known that, without help, this case can be equivalently reduced to either degraded or reversely-degraded case, since the Rx and Ev performance depends on the marginal distributions of $W^n$ and $V^n$, respectively, not on their joint distribution \cite{Bloch-11}. While this is still true for secure Rx help (no Ev help), it is no longer true for non-secure help since Ev performance now depends on both $V^n$ and $W^n$. Thus, the secrecy capacity of this channel can potentially be affected by correlation and does not follow from that of the degraded or reversely-degraded one. Yet, we show below that it is still $C_{s0} + R_h$, irrespective of $r$ (as long as $|r|<1$).

\begin{figure}[t]
    \begin{center}
        \includegraphics[width=3.5in]{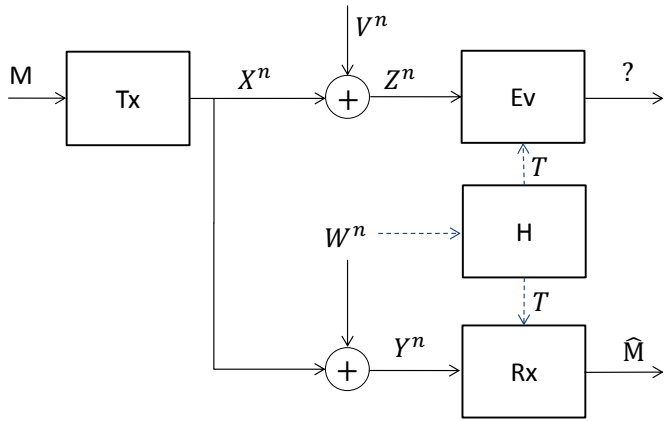}
    \end{center}
%\vspace*{-1\baselineskip}
    \caption{Non-degraded Gaussian wiretap channel with a rate-limited Rx help $T$; noise sequences $W^n$ and $V^n$ are i.i.d. Gaussian (but possibly correlated with each other) and  independent of $M, X^n$; $\sigma_W^2, \sigma_V^2 >0$; $X^n=X^n(M)$, $T=T(W^n)$, $H(T)\le n R_h$.}
    \label{fig.RxH-ND}
\end{figure}

\begin{thm}
\label{thm.Cs.Non-Deg}
Consider the non-degraded Gaussian WTC as in Fig. \ref{fig.RxH-ND} with i.i.d. noise sequences correlated with each other as in \eqref{eq.RWV} and with secure or non-secure Rx help of rate $R_h$; let $\sigma_W^2, \sigma_V^2, P >0$. Its secrecy capacity $C_{s}$ is
\bal
C_{s} =  C_{s0}+ R_h
\eal
for any $|r| <1$.
\end{thm}
\begin{proof}
\textit{Converse:} we consider first the case of secure Rx help (no Ev help). Note that, in this case, Ev's performance depends on $V^n$ only, not on $W^n$; likewise, Rx's performance depends on $W^n$ only, not on $V^n$. Hence, this channel can now be equivalently reduced to degraded or reversely-degraded case, for which the converse have been established in Theorem \ref{thm.Cs} or \ref{thm.Cs.RevDeg}, respectively, so that $R_s \le C_{s0} + R_h$. This argument does not apply for non-secure Rx help. However, since the availability of help to Ev cannot increase the secrecy rate, the same upper bound still holds. This establishes the converse for non-secure Rx help as well.

\textit{Achievability:} Likewise, we can argue that, in the case of secure Rx help, the achievability result of Theorem \ref{thm.Cs} or \ref{thm.Cs.RevDeg} apply. However, it is no longer the case for non-secure help. Furthermore, the achievability under secure help does not imply the achievability under non-secure help. To establish the latter, we use the signaling strategy of Theorem \ref{thm.Cs} if  $\sigma_W^2 < \sigma_V^2$ and of Theorem \ref{thm.Cs.RevDeg} otherwise. Note that no help is used in Phase 1 and hence its signalling rate, as well as $C_{s0}$, are not affected by the correlation of $V^n$ and $W^n$ and therefore the respective results of Theorem \ref{thm.Cs} and  \ref{thm.Cs.RevDeg} apply. Phase 2 rate in \eqref{eq.RxH.Ph2.rate} is not affected by the correlation either since it depends on $W^n$ only (not on $V^n$). To show that secrecy is guaranteed after the two-phase time sharing, we show that the leakage rate $R_{l2}$ of Phase 2 is uniformly bounded for any $\tau$:
\bal
R_{l2} &= n^{-1}I(M_2;Z^n \hat{W}^n|\mathcal{C})\\
%\label{eq.AP.ND.2}
&\le n^{-1}I(M_2;Z^n W^n|\mathcal{C})\\
\label{eq.Rx.ND.corr.3}
&\le n^{-1}I(X^n;Z^n W^n)\\
\label{eq.Rx.ND.corr.4}
&\le I_0(X;Z W)\\
&= h(X+V,W) - h(V,W)\\
&\le h(X+V) + h(W) - h(V,W)\\
\label{eq.Rx.ND.corr.6}
&\le \frac{1}{2}\log\left(1+\frac{P}{\sigma_V^2}\right)- \frac{1}{2}\log(1-r^2)= R_0 < \infty
\eal
where \eqref{eq.Rx.ND.corr.3} follows from Markov chain $(\mathcal{C}, M_2) - X^n - (Z^n,W^n)$; \eqref{eq.Rx.ND.corr.4} holds since the channel is memoryless; \eqref{eq.Rx.ND.corr.6} is due to
\bal
h(X+V)&\le \frac{1}{2}\log\left(2\pi e(P+\sigma_V^2)\right)\\ \notag
h(W,V)&= \frac{1}{2}\log\left((2\pi e)^2 |\bR_{WV}|\right)\\
&= \frac{1}{2}\log\left((2\pi e)^2 \sigma_W^2\sigma_V^2(1-r^2)\right)
\eal
Thus, the overall leakage rate after two-phase time sharing is arbitrarily low, as in \eqref{eq.Rl}, and the secrecy rate in \eqref{eq.Rs.RxH.Deg} is indeed achievable, as required.

\end{proof}

Note that, if $\sigma_W^2 \ge \sigma_V^2$, then $C_{s0}=0$ and $C_s = R_h$, i.e. if the Tx-Rx channel is weaker than the Tx-Ev channel, the secrecy capacity with Rx help is still positive (if $R_h>0$) and independent of $r$ (as long as $|r|<1$), even if the help is not secure. This also holds if $\sigma_W^2 = \sigma_V^2$, unlike the case of the reversely-degraded channel, where $C_s=0$ if $\sigma_W^2 = \sigma_V^2$ and the help is not secure. This is due to $W^n \neq V^n$ in the non-degraded channel (with non-singular correlation) which makes the public "key" $T=\hat{W}^n$ useful to the Rx only, but not to the Ev.

Similarly to the degraded and reversely-degraded WTCs, the same secrecy capacity results even if the helper is aware of the message being transmitted, $T=T(W^n,M)$.

%\newpage
%\vspace*{-.51\baselineskip}
%======================================================================================
\section{The Degraded WTC with Tx Help}
\label{sec.D.Tx}

Let us now consider the setting of Fig. \ref{fig.TxRxEvH} and extend Theorem 1 to the scenario where rate-limited help is available to the Tx, in addition to or instead of the Rx help. Unlike the Rx help case where the causality of help is immaterial (since the Rx starts decoding after the whole block of length $n$ is received), it becomes important for the Tx help setting. Thus, we distinguish between causal Tx help, whereby at time $i$ the Tx help is based on the Rx noise sequence up to time $i$, and non-causal Tx help, whereby the Tx help at time $i=1$ (the very beginning of the transmission) is based on the whole noise sequence $W^n$. Interestingly, the causality of Tx help has no impact on the secrecy capacity (this mimics the respective property of the no-Ev/no-secrecy channel capacity with Tx help in \cite{Lapidoth-20}).

\begin{figure}[t]
    \begin{center}
        \includegraphics[width=3.5in]{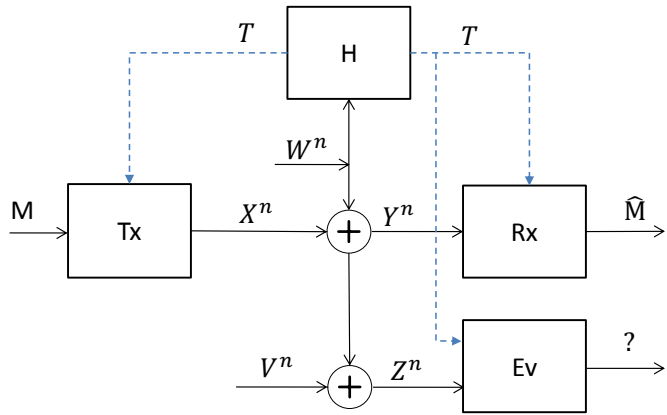}
    \end{center}
    \caption{Degraded wiretap channel with a rate-limited help $T$ at the Tx, Rx and Ev ($T$ is not available to the Ev if the help is secure). $W^n$ and $V^n$ are i.i.d. noise sequences, $\sigma_V^2 >0$; $V^n$ is independent of $W^n,\ X^n,\ M$; $X^n = X^n(M,T)$, $T=T(W^n),\ H(T) \le n R_h$.}
    \label{fig.TxRxEvH}
\end{figure}

\begin{thm}
\label{thm.TxH.Cs}
Consider the degraded Gaussian WTC in Fig. \ref{fig.TxRxEvH} with causal or non-causal Tx help of rate $R_h$, secure or non-secure, in addition to or instead of the same Rx help, and let $\sigma_V^2, P>0$. Its secrecy capacity $C_{s}$ satisfies
\bal
C_{s} \ge C_{s0} + R_h
\eal
where $C_{s0}$ is the secrecy capacity without help. This holds with equality if the help is not secure.
\end{thm}
\begin{proof}
We consider the case of non-secure help, from which the case of secure help follows (since the availability of help to the Ev cannot increase the secrecy rate). The achievability is based on the two-phase flash signaling as in Theorem 1, with noise pre-cancellation at the Tx (a.k.a. dirty-paper coding, as in \cite{Lapidoth-20}) in Phase 2. The converse is based on the functional relationship between the involved random variables as well as the secrecy constraint, in addition to the reliability and power constraints.

\textit{Converse}: we prove the converse when the same non-causal non-secure help $T$ is available to all ends, i.e. the Tx, Rx and Ev as in Fig. \ref{fig.TxRxEvH}. Clearly, the same converse will hold if no Rx help is available or if the help is causal. Using the appropriate Markov chain and functional relationships between the random variables, in addition to the secrecy and reliability constraints, note the following:
\bal
\label{eq.TxH.p1}
n R_s &= H(M)\\
\label{eq.TxH.p2}
&\le H(M|Z^n T) + n\epsilon\\
&= I(M;Y^n|Z^n T) + H(M|Y^nZ^nT) + n\epsilon\\
\label{eq.TxH.p4}
&\le I(M;Y^n|Z^n T) + 2n\epsilon\\
\label{eq.TxH.p5}
&\le I(X^n;Y^n|Z^n T) + 2n\epsilon\\
\label{eq.TxH.p6}
&= I(X^n;Y^n |T) - I(X^n;Z^n|T) +2n\epsilon\\
\label{eq.TxH.p7}
&= h(Y^n|T) - h(Y^n|X^n T) - [h(Z^n|T) - h(Z^n|X^n T)] +2n\epsilon\\
\label{eq.TxH.p8}
&= h(W^n+V^n|T) - h(W^n|T) + h(Y^n|T) - h(Z^n|T) +2n\epsilon\\
\label{eq.TxH.p9}
&\le \frac{n}{2}\log(2\pi e (\sigma_V^2 +\sigma_W^2)) + I(W^n;T) - h(W^n) + h(Y^n|T) - h(Z^n|T) +2n\epsilon\\
\label{eq.TxH.p10}
&= \frac{n}{2}\log\frac{\sigma_V^2 +\sigma_W^2}{\sigma_W^2} + H(T) + h(Y^n|T) - h(Z^n|T) +2n\epsilon\\
\label{eq.TxH.p11}
&\le nR_h + \frac{n}{2}\log\frac{\sigma_V^2 +\sigma_W^2}{\sigma_W^2}\frac{\sigma_W^2 +P}{\sigma_W^2 +\sigma_V^2 + P} +2n\epsilon\\
&= nR_h +  \frac{n}{2}\log\left(1+\frac{P}{\sigma_W^2}\right) - \frac{n}{2}\log\left(1+\frac{P}{\sigma_V^2+\sigma_W^2}\right) +2n\epsilon\\
\label{eq.TxH.p13}
&= n(R_h + C_{s0} + 2\epsilon)
\eal
where \eqref{eq.TxH.p2} follows from the secrecy constraint $I(M;Z^n T) \le n\epsilon$; \eqref{eq.TxH.p4} follows from Fano inequality (due to the reliability constraint) $H(M|Y^nZ^n T) = H(M|Y^n T) \le n \epsilon$;
\eqref{eq.TxH.p5} and \eqref{eq.TxH.p6} follow from Markov chain $M - X^n - Y^n - Z^n$ conditional on $T$; \eqref{eq.TxH.p8} is due to the independence of $X^n$ and $(W^n,V^n)$ conditional on $T$;
\eqref{eq.TxH.p9} follows since conditioning cannot increase entropy;
\eqref{eq.TxH.p10} is due to $I(W;T)=H(T)$; \eqref{eq.TxH.p11} follows from Lemma \ref{lemma.DH} below. Since \eqref{eq.TxH.p13} holds for any $\epsilon >0$, it follows that $C_s \le C_{s0} + R_h$, as desired. Clearly, the same inequality holds if $T$ is not available to the Rx.
\begin{lemma}
\label{lemma.DH}
The following inequality holds in the considered setting:
\bal
\Delta h = h(Y^n|T) - h(Z^n|T) \le \frac{n}{2}\log\frac{\sigma_W^2 +P}{\sigma_W^2 +\sigma_V^2 + P}
\eal
\begin{proof}
It has been proved in \cite[eq. (46)]{Lapidoth-20} that
\bal
\label{eq.HYT.ineq}
h(Y^n|T) \le \frac{n}{2}\log(2\pi e(\sigma_W^2 +P))
\eal
(the proof is not trivial since $X^n$ and $W^n$ are \textit{not} independent, due to help $T=T(W^n)$). To bound $h(Z^n|T)$ likewise, note that
\bal
\label{eq.L1.hZnT}
h(Z^n|T) = \sum_t p_T(t) h(Y^n+V^n|T=t)
\eal
where $p_T(t)$ is the distribution of $T$. Using the entropy power inequality
\bal
2^{\frac{2}{n}h(Y^n+V^n|T=t)} \ge 2^{\frac{2}{n}h(Y^n|T=t)} + 2^{\frac{2}{n}h(V^n|T=t)}
\eal
it follows that
\bal
h(Y^n+V^n|T=t) \ge \frac{n}{2} \log(2^{\frac{2}{n} h(Y^n|T=t)} + 2\pi e \sigma_V^2)
\eal
and hence
\bal
h(Z^n|T) &\ge \frac{n}{2} \log\left(2^{\frac{2}{n}\sum_t p_T(t) h(Y^n|T=t)} + 2\pi e \sigma_V^2\right)\\ \notag
 &= \frac{n}{2} \log\left(2^{\frac{2}{n}h(Y^n|T)} + 2\pi e \sigma_V^2\right)
\eal
where the inequality is due to the convexity of the log-sum-exp function \cite[p. 72]{Boyd-04}. Finally,
\bal
\Delta h &\le h(Y^n|T) - \frac{n}{2} \log\left(2^{\frac{2}{n}h(Y^n|T)} + 2\pi e \sigma_V^2\right)\\
&\le \frac{n}{2}\log(2\pi e(\sigma_W^2 +P))  - \frac{n}{2} \log\left(2^{\log(2\pi e(\sigma_W^2 +P))} + 2\pi e \sigma_V^2\right)\\
\label{eq.L1.dh}
&= \frac{n}{2}\log\frac{\sigma_W^2 +P}{\sigma_W^2 +\sigma_V^2 + P}
\eal
as required, where the inequality is due to \eqref{eq.HYT.ineq} and $f(x) = x - \log(2^x + c)$ being an increasing function of $x$ for any $c>0$.
\end{proof}
\end{lemma}

\begin{figure}[t]
    \begin{center}
        \includegraphics[width=4.5in]{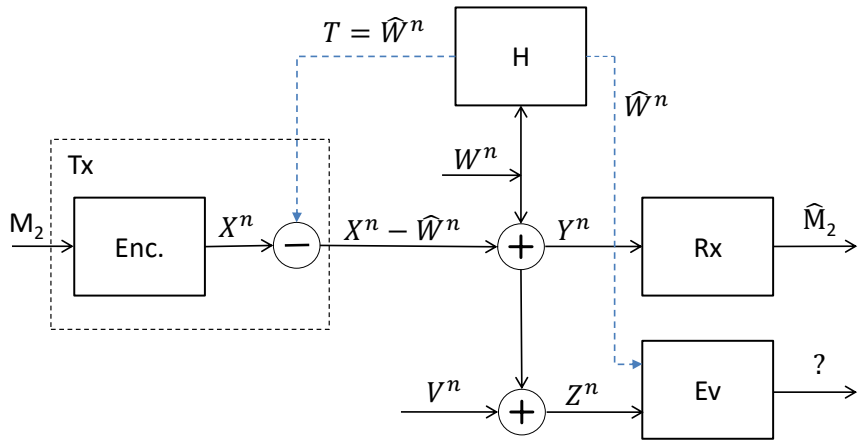}
    \end{center}
    \caption{Phase 2 signalling for the degraded WTC: the causal help $T$ is a scalar-quantized noise $\hat{W}^n$ , $\hat{W}_i=Q(W_i)$, pre-subtracted at the Tx; $X^n=X^n(M_2)$ is a codeword from i.i.d.-generated codebook, as in \cite{Lapidoth-20}.}
    \label{fig.TxEvH-A}
\end{figure}

\textit{Achievability}:
We consider the case of causal help being available to the Tx and Ev but not to the Rx. This will also establish achievability when the same help is also available to the Rx or/and when Tx help is non-causal (since adding Rx help or removing causality constraint cannot decrease achievable rates). Similarly to Theorem 1, we use a two-phase signalling, where Phase 1 of duration $(1-\tau)$ makes use of no-help regular wiretap coding and thus achieves the secrecy rate $C_{s0} - \epsilon$ for any $\epsilon >0$. Phase 2 of duration $\tau$ is the same as in \cite{Lapidoth-20}, which makes use of regular (no-wiretap) coding and pre-substraction of the scalar-quantized noise (available via the rate-limited help link) at the Tx, as shown in Fig. \ref{fig.TxEvH-A}:
\bal\notag
Y_i &= X_i - \hat{W}_i + W_i\\
Z_i &= Y_i + V_i
\eal
where $X^n = X^n(M_2)$ using i.i.d.-generated codebook $\mathcal{C}$, $T=\hat{W}^n$ is a scalar-quantized noise, $\hat{W}_i=Q(W_i)$, where the quantizer uses $L=\lfloor2^{R_h/\tau}\rfloor $ levels for each sample, which require the average rate $\tau \log(L) \le R_h$ to be transmitted over the help link\footnote{An alternative Phase 2 strategy using a simple lattice code with a uniform scalar quantizer is proposed in \cite{Merhav-21}.}. For further use, note that $V^n\ \bot\ (W^n, \hat{W}^n, X^n, M_2)$ and $(W^n, \hat{W}^n)\ \bot\ (V^n, X^n, M_2)$, where $\bot$ means statistical independence, so that the following Markov chains hold:
\bal
(M_2, \mathcal{C}) - X^n - Y^n  - Z^n;\ (M_2, \mathcal{C}) - X^n - (Z^n, W^n, \hat{W}^n)
\eal
Following \cite[eq. (24)]{Lapidoth-20}, this Phase 2 signalling achieves the rate arbitrary close to
\bal
\label{eq.TxH.Ph2.rate}
\frac{R_h}{\tau} + \frac{1}{2}\log\left(2^{-2R_h/\tau} + \alpha_W P (1-2^{-R_h/\tau})^2 [1+o(1)]\right) = \frac{R_h}{\tau} [1+o(1)]
\eal
where $\alpha_W=2(\pi\sqrt{3}\sigma_W^2)^{-1}$ and $o(1)\to 0$ as as $\tau \to 0$. Thus, the overall two-phase signalling rate (after time sharing) is
\bal
\label{eq.Rs.TxH.Deg}
(1-\tau)(C_{s0} - \epsilon) &+ \tau R_h/\tau(1+o(1)) \to C_{s0} +R_h -\epsilon
\eal
for any $\epsilon >0$, as $\tau \to 0$.

It remains to show that this rate is indeed secure, i.e. the information leakage rate to the Ev is arbitrary small. This is clearly the case in Phase 1 since regular wiretap coding is used in this phase so that its leakage rate is $R_{l1} = n^{-1} I(M_1; Z^n) \le \delta$ for any $\delta >0$ and sufficiently-large $n$.
To see that secrecy is guaranteed after two-phase time sharing (even though no wiretap coding is used in Phase 2), we show that Phase 2 leakage rate is uniformly bounded for any $\tau$:
\bal
R_{l2} &= n^{-1}I(M_2;Z^n \hat{W}^n|\mathcal{C})\\
\label{eq.TH.2}
&\le n^{-1}I(X^n;Z^n \hat{W}^n|\mathcal{C})\\
\label{eq.TH.3}
&\le n^{-1}I(X^n;Z^n\hat{W}^n)\\
\label{eq.TH.4}
&\le n^{-1}I(X^n;Z^n W^n)\\
\label{eq.TH.5}
&\le I_0(X;Z W)\\
&= I_0(X;X+V)\\
&\le \frac{1}{2}\log\left(1+\frac{P}{\sigma_V^2}\right) = C_2 < \infty
\eal
where \eqref{eq.TH.2} is due to Markov chain $M_2 - X^n - Z^n \hat{W}^n$; \eqref{eq.TH.3} is due to Markov chain $\mathcal{C} - X^n - Z^n \hat{W}^n$; \eqref{eq.TH.4} is due to $\hat{W}^n = Q(W^n)$; \eqref{eq.TH.5} holds since the channel is memoryless; $I_0$ is the mutual information induced by input $X$ with the distribution $p_0(x)=n^{-1}\sum_i p_{x_i}(x)$.

Thus, the overall leakage rate after two-phase time sharing is
\bal
\label{eq.TH.Rl}
R_l= (1-\tau)R_{l1} + \tau R_{l2} \le (1-\tau)\delta + \tau C_2 \to \delta
\eal
as $\tau \to 0$, for any $\delta >0$, as required.
\end{proof}

Note that the availability of the Rx help, in addition to the Tx help, does not increase the secrecy capacity (provided the help $T$ is the same in both cases) so that one link can be omitted without affecting the capacity.

Similarly to the Rx help case, if $\sigma_V^2=0$ and the Tx (or joint Tx/Rx) help is not secure, then the secrecy capacity is zero, since the Ev has access to exactly the same information as the Rx so that no secrecy is possible, i.e. $C_s(\sigma_V^2)$ is a discontinuous function at $\sigma_V^2 =0$:
\bal
\lim_{\sigma_V^2 \to 0^+} C_s(\sigma_V^2) = R_h >0
\eal
while $C_s(0)=0$.

%\vspace*{-.51\baselineskip}
%======================================================================================
\section{The Reversely-Degraded WTC with Tx Help}
\label{sec.Rev-Deg.TxH}

Let us consider the reversely-degraded Gaussian WTC, as in Fig. \ref{fig.TxRxEvH-RD}, with Tx help, in addition to or instead of the Rx help ($T$ is not available to the Ev if help is secure). While its secrecy capacity is zero without help, this is not the case when help is present, even if it is not secure, as the following Theorem shows.

\begin{figure}[t]
    \begin{center}
        \includegraphics[width=3.5in]{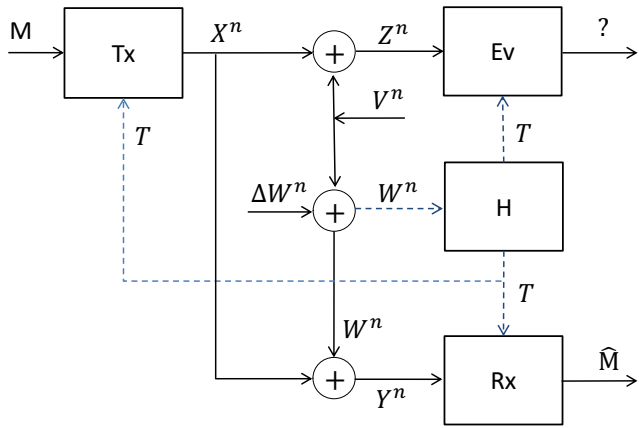}
    \end{center}
    \caption{Reversely-degraded wiretap channel with a rate-limited help $T$ at the Tx, Rx and Ev. $\Delta W^n$ and $V^n$ are i.i.d. noise sequences, $\sigma_V^2,\sigma_{\Delta W}^2 >0$; $V^n$, $\Delta W^n$ and $M$ are independent of each other; $X^n = X^n(M,T)$, $T=T(W^n),\ H(T) \le n R_h$.}
    \label{fig.TxRxEvH-RD}
\end{figure}

\begin{thm}
\label{thm.TxH.Cs.RD}
Consider the reversely-degraded Gaussian WTC with causal or non-causal Tx help of rate $R_h$, secure or not, in addition to or instead of the same Rx help, as in Fig. \ref{fig.TxRxEvH-RD}, and let $\sigma_V^2, \sigma_{\Delta W}^2, P>0$. Its secrecy capacity $C_{s}$ satisfies
\bal
C_{s} \ge R_h
\eal
and this holds with equality if the help is not secure.
\end{thm}
\begin{proof}
It is sufficient to consider the case of non-secure help since the case of secure one follows from it.

\textit{Converse}: we prove the converse when the same (non-secure non-causal) help $T$ is available to all ends, i.e. the Tx, Rx and Ev as in Fig. \ref{fig.TxRxEvH-RD}. Clearly, the same converse will hold if no Rx help is available or if the help is causal. First, note that \eqref{eq.TxH.p1}-\eqref{eq.TxH.p5} still hold, since channel degradedness plays no role there. Therefore,
\bal\notag
n R_s &= H(M)\\
\label{eq.RD.TxH.p2}
&\le I(X^n;Y^n|Z^n T) + 2n\epsilon\\
\label{eq.RD.TxH.p3}
&= I(X^n;X^n+V^n+\Delta W^n|X^n+V^n,T) +2n\epsilon\\
\label{eq.RD.TxH.p4}
&= I(V^n;\Delta W^n|X^n+V^n,T) +2n\epsilon\\
\label{eq.RD.TxH.p5}
&= h(\Delta W^n|X^n+V^n,T)- h(\Delta W^n|X^n,V^n,T) +2n\epsilon\\
\label{eq.RD.TxH.p6}
&\le h(\Delta W^n|T)- h(\Delta W^n|V^n,T) +2n\epsilon\\
\label{eq.RD.TxH.p7}
&\le h(\Delta W^n|T)- h(\Delta W^n) +H(T) +2n\epsilon\\
\label{eq.RD.TxH.p8}
&\le H(T) + 2n\epsilon
\eal
where \eqref{eq.RD.TxH.p6} follows from Markov chain $X^n - T - (V^n,\Delta W^n)$ (i.e., conditional independence of $X^n$ and $(V^n,\Delta W^n)$ given $T$) and since conditioning cannot increase the entropy; \eqref{eq.RD.TxH.p7} is due to
\bal
h(\Delta W^n|V^n,T) \ge h(\Delta W^n)- H(T)
\eal
which in turn follows from
\bal
I(\Delta W^n;T|V^n) &= h(\Delta W^n|V^n) - h(\Delta W^n|V^n,T)\\
\label{eq.RD.TxH.p11}
 &= h(\Delta W^n) - h(\Delta W^n|V^n,T)\\
\label{eq.RD.TxH.p12}
 &\le H(T)
\eal
where \eqref{eq.RD.TxH.p11} is due to the independence of $\Delta W^n$ and $V^n$. Since \eqref{eq.RD.TxH.p8} holds for any $\epsilon>0$, it follows that $R_s \le n^{-1}H(T) = R_h$, as required.

\begin{figure}[t]
    \begin{center}
        \includegraphics[width=4.5in]{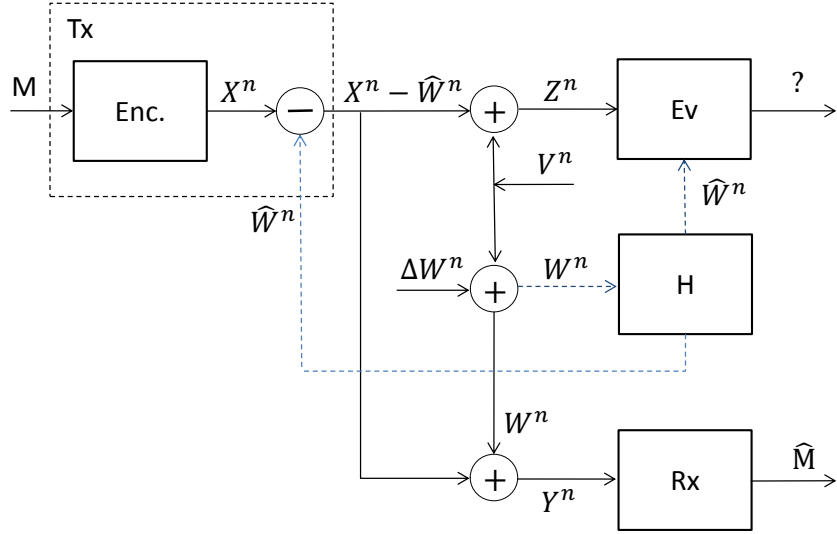}
    \end{center}
    \caption{Phase 2 signalling for the reversely-degraded WTC: the causal help $T=\hat{W}^n$ is a scalar-quantized noise, $\hat{W}_i=Q(W_i)$, pre-subtracted at the Tx; $X^n=X^n(M_2)$ is a codeword from i.i.d.-generated codebook.}
    \label{fig.TxEvH-A-RD}
\end{figure}

\textit{Achievability}: We consider the case where causal help $T=\hat{W}^n$ is a scalar-quantized noise, $\hat{W}_i=Q(W_i)$, available to the Tx and Ev but not to the Rx. Two-phase transmission is used again, where nothing is transmitted in Phase 1 and regular (no wiretap coding) flash signalling is used in Phase 2, where the latter achieves the rate as in \eqref{eq.TxH.Ph2.rate} so that, after time sharing, the achieved rate is arbitrary close to
\bal
R_s = \tau \frac{R_h}{\tau}(1+o(1)) \to R_h
\eal
as $\tau \to 0$. To see that this rate is indeed secure after the time-sharing (i.e., the information leakage rate is arbitrary-low), we show that the Phase 2 leakage rate is uniformly bounded:
\bal
R_{l2} &= n^{-1}I(M_2;Z^n \hat{W}^n|\mathcal{C})\\
\label{eq.RD.TH-C.2}
&\le n^{-1}I(X^n;Z^n \hat{W}^n|\mathcal{C})\\
\label{eq.RD.TH-C.3}
&\le n^{-1}I(X^n;Z^n\hat{W}^n)\\
\label{eq.RD.TH-C.4}
&= n^{-1}I(X^n;\hat{W}^n) + n^{-1}I(X^n;Z^n|\hat{W}^n)\\
\label{eq.RD.TH-C.5}
&= n^{-1}I(X^n;X^n -\hat{W}^n +V^n|\hat{W}^n)\\
\label{eq.RD.TH-C.6}
&= n^{-1}h(X^n +V^n|\hat{W}^n) - n^{-1}h(V^n|\hat{W}^n,X^n)\\
\label{eq.RD.TH-C.7}
&\le n^{-1}h(X^n +V^n) - n^{-1}h(V^n|\hat{W}^n)\\
\label{eq.RD.TH-C.8}
&\le \frac{1}{2}\log\left(1+\frac{P}{\sigma_V^2}\right) + \frac{1}{2}\log\left(1+\frac{\sigma_V^2}{\sigma_{\Delta W}^2}\right)= C_2' < \infty
\eal
where \eqref{eq.RD.TH-C.2} is due to Markov chain $M_2 - X^n - Z^n \hat{W}^n$; \eqref{eq.RD.TH-C.3} is due to Markov chain $\mathcal{C} - X^n - Z^n \hat{W}^n$; \eqref{eq.RD.TH-C.5} and \eqref{eq.RD.TH-C.7} are due to the independence of $X^n$ and $(V^n, \hat{W}^n)$; \eqref{eq.RD.TH-C.8} follows from
\bal
h(X^n +V^n)&\le \frac{n}{2}\log(2\pi e(P+\sigma_V^2)),\\
h(V^n|\hat{W}^n) &\ge h(V^n|W^n)\\
&=h(V^n) - h(W^n) + h(\Delta W^n)\\
&= \frac{n}{2}\log(2\pi e\sigma_V^2) - \frac{n}{2}\log(2\pi e(\sigma_{\Delta W}^2+\sigma_V^2)) + \frac{n}{2}\log(2\pi e\sigma_{\Delta W}^2)
\eal
Thus, after time-sharing, which is equivalent here to Phase 2 only signaling, the leakage rate is
\bal
R_l = \tau R_{l2} \le \tau C_2' \to 0
\eal
when $\tau \to 0$, as required.
\end{proof}

We remark that, as in the reversely-degraded WTC with Rx help, no wiretap coding is needed here to achieve its secrecy capacity if the help is not secure. Burst signalling alone (with regular coding) is sufficient and arbitrarily low leakage rate can be achieved by reducing signaling interval $\tau$. The presence of help $T$ at the Rx, in addition to the Tx, does not increase the capacity. Even though the help is not secure, it still boosts significantly the secrecy capacity, which is zero without help. This is so since the help $T$ serves here as a public key: even though this key is available to the Ev, it cannot make use of it since it does not have the right "lock".

Similarly to the reversely-degraded WTC with Rx help, $C_s=0$ if $\sigma_{\Delta W}^2=0$ and help is not secure (since the Ev receives the same information as the Rx so that no secrecy is possible) and therefore $C_s(\sigma_{\Delta W}^2)$ is discontinuous at $\sigma_{\Delta W}^2=0$:
\bal
C_s(\sigma_{\Delta W}^2) = R_h >0\ \forall\ \sigma_{\Delta W}^2>0
\eal
while $C_s(0)=0$, for any $R_h >0$, i.e. more noise at the Rx ($\sigma_{\Delta W}^2>0$) is better for the secrecy capacity of this channel.

%\vspace*{-.51\baselineskip}
%======================================================================================
\section{The Non-Degraded WTC with Tx Help}
\label{sec.Non-Deg.TxH}

Let us now consider the non-degraded Gaussian wiretap channel where $W^n$ and $V^n$ are i.i.d. noise sequences, possibly correlated with each other as in \eqref{eq.RWV}, see Fig. \ref{fig.TxRxH-ND} (if help is secure, $T$ is not available to the Ev). Similarly to the case of Rx help, this channel cannot be equivalently reduced to degraded or reversely-degraded case when help is present (even if it is secure). Its secrecy capacity is characterized below.

\begin{figure}[t]
    \begin{center}
        \includegraphics[width=3.5in]{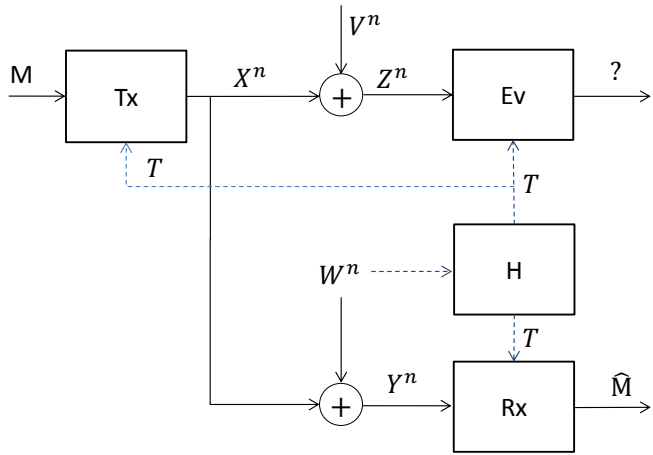}
    \end{center}
    \caption{Non-degraded wiretap channel with a rate-limited help $T$ at the Tx, Rx and Ev. $W^n$ and $V^n$ are i.i.d. noise sequences (possibly correlated with each other) independent of $M$; $\sigma_W^2, \sigma_V^2 >0$; $X^n = X^n(M,T)$, $T=T(W^n),\ H(T) \le n R_h$; conditional on $T$, $X^n$ is independent of $W^n, V^n$.}
    \label{fig.TxRxH-ND}
\end{figure}

\begin{prop}
\label{prop.ND.Cs}
Consider the non-degraded Gaussian WTC channel with secure or non-secure Tx help of rate $R_h$, causal or non-causal, in addition to or instead of the same Rx help, as in Fig. \ref{fig.TxRxH-ND}, where the noise sequences $W^n$ and $V^n$ are i.i.d. but possibly correlated with each other as in \eqref{eq.RWV} and $\sigma_W^2, \sigma_V^2, P >0$. Its secrecy capacity can be lower bounded as follows:
\bal
\label{eq.ND.Cs}
C_s \ge C_{s0} + R_h
\eal
for any $|r| <1$.
\end{prop}
\begin{proof}
To show the achievability of $C_{s0} +R_h$, we use the same two-phase signalling as in Theorem \ref{thm.TxH.Cs}, where Phase 1 makes use of the standard wiretap codes and no help and thus achieves the secrecy rate arbitrary close to $C_{s0}$. Note that $C_{s0}$ is not affected by the correlation since it depends on the marginal distributions of $W^n$ and $V^n$, not on the joint one, so that Phase 1 secrecy rate is not affected either. Likewise, Phase 2 makes use of standard (not wiretap) codes and pre-subtracts quantized noise at the Tx, as in Fig. \ref{fig.TxEvH-A}, and achieves the rate as in  \eqref{eq.TxH.Ph2.rate} (regardless of the correlation), so that, after the time sharing, the rate is as in \eqref{eq.Rs.TxH.Deg}.  To show that this rate is indeed secure, we show that Phase 2 leakage rate is uniformly bounded for any $\tau$. To this end, note that \eqref{eq.TH.2}-\eqref{eq.TH.5} still hold since channel degradedness or noise correlation play no role there so that
\bal
R_{l2} &\le I_0(X;Z W)\\
\label{eq.ND-C.TH.2}
&= I_0(X;X+V,W)\\
\label{eq.ND-C.TH.3}
&\le h(X+V) +h(W) - h(V,W)\\
\label{eq.ND-C.TH.4}
&\le \frac{1}{2}\log\left(1+\frac{P}{\sigma_V^2}\right) - \frac{1}{2}\log\left(1-r^2\right)= C_2 < \infty
\eal
where \eqref{eq.ND-C.TH.3} holds since $X$ is independent of $W,V$; \eqref{eq.ND-C.TH.4} holds since
\bal
h(X+V)&\le \frac{1}{2}\log\left(2\pi e(P+\sigma_V^2)\right)\\ \notag
h(W,V)&= \frac{1}{2}\log\left((2\pi e)^2 |\bR_{WV}|\right)\\
&= \frac{1}{2}\log\left((2\pi e)^2 \sigma_W^2\sigma_V^2(1-r^2)\right)
\eal
Therefore, the overall leakage rate after the two-phase time sharing is arbitrary small as in \eqref{eq.TH.Rl}, as required, and the achieved rate in \eqref{eq.Rs.TxH.Deg} is indeed secure.
\end{proof}

\begin{figure}[t]
    \begin{center}
        \includegraphics[width=4.5in]{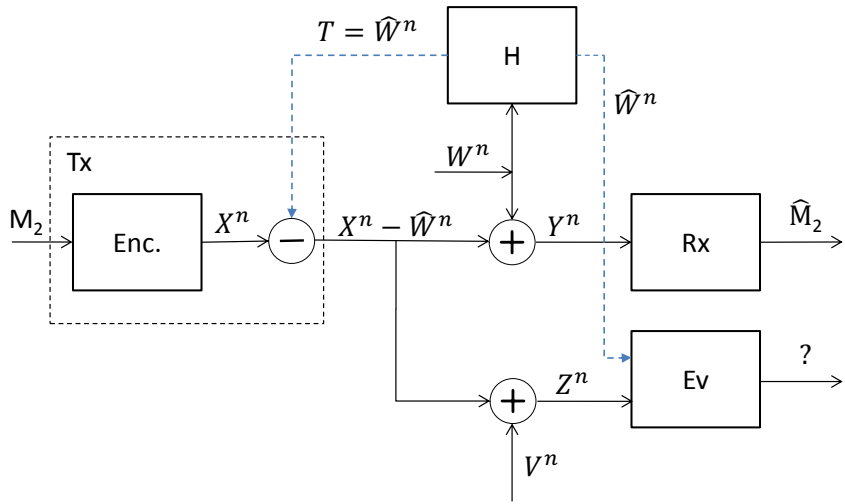}
    \end{center}
    \caption{Phase 2 signalling for the non-degraded WTC: $\hat{W}^n=Q(W^n)$ is scalar-quantized noise, pre-subtracted at the Tx; $X^n=X^n(M_2)$ is a codeword from i.i.d.-generated codebook.}
    \label{fig.ND.TxEvH-A}
\end{figure}

Thus, the Tx help of rate $R_h$, secure or non-secure, causal or non-causal, brings in the secrecy capacity boost of at least $R_h$ in this configuration, regardless of the correlation (as long as $|r|<1$). It is an open question whether $C_s=C_{s0}+R_h$.

%\vspace*{-.51\baselineskip}
%======================================================================================
\section{Independent Tx/Rx Help Links}
\label{sec.Ind.TxRxH}

In the preceding sections, we have considered the scenarios where the same help was available at the Tx and Rx and have shown that the presence of Tx help in addition to the same Rx help (or vice versa) has no impact on the secrecy capacity and therefore one link can be removed without affecting the capacity.

One may wonder whether this still holds if help links are not identical. Therefore, we consider the scenario whereby independent help links are  available to the Tx and Rx of rate $R_{h1}$ and $R_{h2}$, respectively. The total help $T$ is composite: $T=(T_1,T_2)$, where $T_1$ is available to the Tx and $T_2$  - to the Rx while the whole help $T$ is available to the Ev (in the case of non-secure help); $T_1, T_2$ are independent of each other (e.g. based on different parts of the i.i.d. noise sequence $W^n$) and $H(T_k) \le nR_{hk},\ k=1,2$, so that
\bal
H(T) = H(T_1) + H(T_2) \le R_{h1} +R_{h2} = R_h
\eal

We consider first the degraded WTC as in Fig. \ref{fig.TxRxEvH-I}.

\begin{figure}[t]
    \begin{center}
        \includegraphics[width=3.5in]{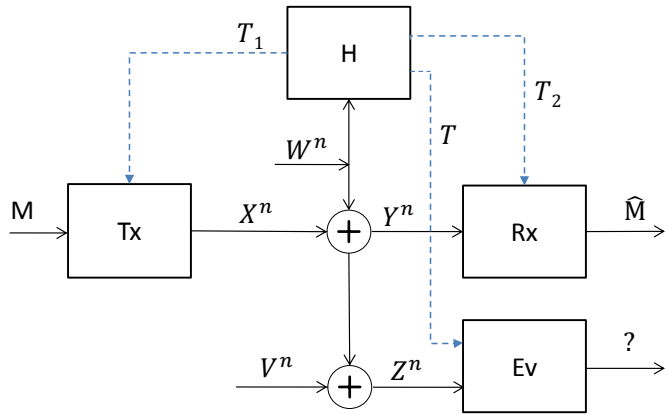}
    \end{center}
    \caption{Degraded wiretap channel with independent help links to the Tx and Rx of rate $R_{h1}$ and $R_{h2}$, respectively; the total help $T=(T_1,T_2)$ is available to the Ev (if help is not secure). $W^n$ and $V^n$ are i.i.d. noise sequences; $P, \sigma_V^2 >0$; $V^n$ is independent of $W^n,\ X^n,\ M$; $X^n = X^n(M,T_1)$, $T=T(W^n),\ H(T_k) \le n R_{hk}$.}
    \label{fig.TxRxEvH-I}
\end{figure}

\begin{thm}
\label{thm.TxRxH-I.Cs}
Consider the degraded Gaussian WTC with causal or non-causal Tx help of rate $R_{h1}$ and Rx help of rate $R_{h2}$ independent of each other (secure or not) as in Fig. \ref{fig.TxRxEvH-I}, and let $\sigma_V^2, P>0$. Its secrecy capacity $C_{s}$ satisfies
\bal
C_{s} \ge C_{s0} + R_{h1} + R_{h2}
\eal
where $C_{s0}$ is the secrecy capacity without help. This holds with equality if help is not secure.
\end{thm}
\begin{proof}
To prove achievability, we use three-phase signalling combining Rx and Tx help in independent phases:

\begin{enumerate}
  \item Phase 1 of duration $(1-\tau_1-\tau_2)$:  the standard wiretap coding is used without any help, as in Theorems \ref{thm.Cs}, \ref{thm.TxH.Cs}.
  \item Phase 2 of duration $\tau_1$:  the same as for the Tx help in Theorem \ref{thm.TxH.Cs} (flash signalling with Tx help and regular coding); no Rx help is used in this phase.
  \item Phase 3 of duration $\tau_2$:  the same as for the Rx help in Theorem \ref{thm.Cs} (flash signalling with Rx help and regular coding); no Tx help is used in this phase.
\end{enumerate}

Clearly, a secrecy rate arbitrary close to $C_{s0}$ is achievable in Phase 1, as before. Likewise, based on Theorems \ref{thm.Cs} and \ref{thm.TxH.Cs}, secrecy rates arbitrary close to $R_{h1}/\tau_1(1+o(1))$ and $R_{h2}/\tau_2(1+o(1))$ are achievable in Phases 2 and 3, and, after three-phase time sharing, a secrecy rate arbitrary close to
\bal
(1-\tau_1-\tau_2)C_{s0} + (R_{h1}+ R_{h2})(1+o(1)) \to C_{s0} + R_{h1} + R_{h2}
\eal
is achievable as $\tau_1,\tau_2 \to 0$.

\textit{Converse}: a key observation here is that the converse of Theorem \ref{thm.TxH.Cs} still holds with $T=(T_1,T_2)$. Indeed, \eqref{eq.TxH.p1} - \eqref{eq.TxH.p13} do hold, where \eqref{eq.TxH.p4} holds since
\bal
H(M|Y^nZ^nT)\le H(M|Y^n T_2) \le n\epsilon
\eal
\eqref{eq.TxH.p11} holds since Lemma 1 still holds, due to
\bal
%\label{eq.HYT.ineq}
h(Y^n|T) \le h(Y^n|T_1)\le \frac{n}{2}\log(2\pi e(\sigma_W^2 +P))
\eal
where the last inequality is due to \eqref{eq.HYT.ineq}, and since \eqref{eq.L1.hZnT} - \eqref{eq.L1.dh} do hold with $T=(T_1,T_2)$.
\end{proof}

Note that if Tx/Rx help links are independent of each other, the secrecy capacity boost is their combined rate $R_{h1}+R_{h2}$, unlike the case of identical Tx/Rx help where the boost is just an individual help rate, as in Theorem \ref{thm.TxH.Cs}, and one of the two help links can be removed without any effect on the capacity.

Next, we consider  the reversely-degraded channel with independent help links as in Fig. \ref{fig.TxRxEvH-I-RD}, where $T_1$ and $T_2$ are independent of each other.

\begin{figure}[t]
    \begin{center}
        \includegraphics[width=3.5in]{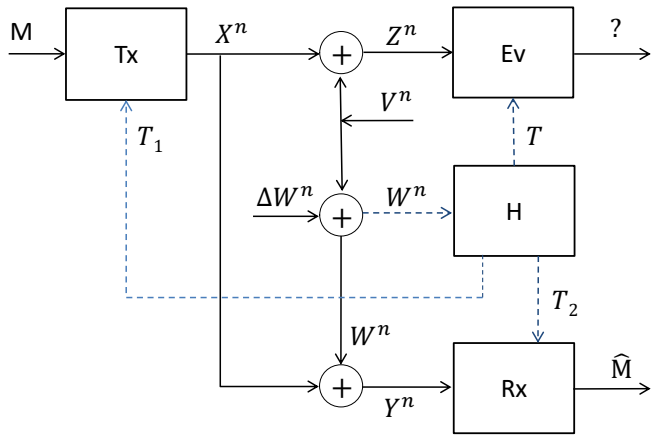}
    \end{center}
    \caption{Reversely-degraded wiretap channel with independent help links to the Tx and Rx of rate $R_{h1}$ and $R_{h2}$, respectively; the total help $T=(T_1,T_2)$ is available to the Ev (if help is not secure). $\Delta W^n$ and $V^n$ are i.i.d. noise sequences independent of each other; $\sigma_{\Delta W}^2, \sigma_V^2, P >0$; $X^n = X^n(M,T_1)$, $T=T(W^n),\ H(T_k) \le n R_{hk}$.}
    \label{fig.TxRxEvH-I-RD}
\end{figure}

\begin{thm}
\label{thm.TxRxH-RD-I.Cs}
Consider the reversely-degraded Gaussian WTC with independent causal or non-causal help links as in  Fig. \ref{fig.TxRxEvH-I-RD}, and let $\sigma_{\Delta W}^2, \sigma_V^2, P>0$. Its secrecy capacity $C_{s}$ satisfies
\bal
C_{s} \ge R_{h1} + R_{h2}
\eal
This holds with equality if the help is not secure.
\end{thm}
\begin{proof}
The achievability makes use of the three-phase signaling as in Theorem \ref{thm.TxRxH-I.Cs} where nothing is transmitted in Phase 1. The converse is established by observing that \eqref{eq.RD.TxH.p2}-\eqref{eq.RD.TxH.p12} still hold with $T=(T_1,T_2)$.
\end{proof}

For the non-degraded channel with independent Tx/Rx help links, it can be shown, in a similar way, that Proposition \ref{prop.ND.Cs} still holds with $R_h = R_{h1} + R_{h2}$.

Thus, in all considered configurations, the independent Tx/Rx help links provide additive boost $R_{h1}+R_{h2}$ in secrecy rates, unlike the same help links whereby one link can be omitted without affecting the capacity. This mimics the respective property of the no-Ev channel with independent help links in \cite{Marti-19}.

Finally, one may envision the case of composite help $T=(T_1,T_2)$ where $T_1, T_2$ are not independent of each other but are not identical either (which may be due to e.g. certain limitations in the system architecture). In this case and if the help is not secure, it is not difficult to see that Theorems \ref{thm.TxRxH-I.Cs} and \ref{thm.TxRxH-RD-I.Cs} still hold with
\bal
C_{s0} + \max\{R_{h1},R_{h2}\} \le C_s \le C_{s0} +R_{h}
\eal
where $C_{s0}=0$ for the latter, and $H(T) \le n R_h$, so that the boost in the secrecy capacity is at least $\max\{R_{h1},R_{h2}\}$ (and this lower bound is achievable with one help link only). It remains to be seen whether the upper bound is achieved with equality, i.e. whether the boost is actually $R_h > \max\{R_{h1},R_{h2}\}$.

%========================================================================
\section{Conclusion}

The SISO Gaussian wiretap channel with rate-limited help at the receiver (decoder) or/and the transmitter (encoder) was studied and its secrecy capacity has been established under various channel configurations (degraded, reversely degraded and non-degraded) for secure and non-secure help. In all considered cases but one, the rate-limited help results in the secrecy capacity boost (compared to the standard "no help" case) equal to the help rate, so that positive secrecy rate is achievable even for reversely-degraded channel, where the secrecy capacity is zero without help. Surprisingly, secure Rx help does not result in higher capacity compared to non-secure one and more noise at the legitimate receiver can sometimes be beneficial for secrecy capacity. When Tx and Rx help links are identical (carry the same help), any one can be removed without affecting the capacity. However, when the help links are independent, the boost in secrecy capacity equals to the sum of help rates and no one link can be omitted without loss in the capacity. Non-singular Rx and Ev noise correlation has no impact on the secrecy capacity. In the case of Rx help, secure or non-secure, the secrecy capacity is not increased even if the helper is aware of the message being transmitted. In the case of non-secure Tx help, non-causal help does not bring in any increase in the secrecy capacity over the causal one.
Comparing the above results to those for the no-Ev channel with help in \cite{Bross-20}-\cite{Marti-19}, we conclude that the boost in capacity equal to the help rate comes with secrecy "for free". It remains to be seen whether the secrecy of Tx help or helper's knowledge of the message brings in any increase in the secrecy capacity.


\begin{thebibliography}{1}

\bibitem{Bloch-11} M. Bloch, J. Barros, Physical-Layer Security: From Information Theory   to Security Engineering, Cambridge University Press, 2011.
\bibitem{Schaefer-17}  R.F. Schaefer, H. Boche, A. Khisti, H.V. Poor (Eds.), Information Theoretic Security and Privacy of Information Systems, Cambridge University Press, 2017.
\bibitem{Regalia-15} P. A. Regalia et al (Eds.), Secure Communications via Physical-Layer and Information-Theoretic Techniques, Proceedings of the IEEE, vol.103, no.10, Oct. 2015.

\bibitem{Wu-18} Y. Wu et al., A Survey of Physical Layer Security Techniques for 5G Wireless Networks and Challenges Ahead, IEEE JSAC, vol. 36, no. 4, pp. 679-695, Apr. 2018.
\bibitem{Mucchi-21} L. Mucchi et al., Physical-Layer Security in 6G Networks, IEEE Open J. Comm. Society, vol. 2, pp. 1901-1914, 2021.
\bibitem{Chorti-22} A. Chorti et al., Context-Aware Security for 6G Wireless: The Role of Physical Layer Security, IEEE Comm. Standards Magazine, vol. 6, no. 1, pp. 102-108, Mar. 2022.

\bibitem{Shannon-49} C. E. Shannon, Communication Theory of Secrecy Systems, Bell Syst. Tech. J., vol. 28, pp. 656--715, Oct. 1949.
\bibitem{Wyner-75} A.D. Wyner, The Wire-Tap Channel, Bell System Technical Journal, v. 54, no. 8, pp. 1355--1387, Oct. 1975.
\bibitem{Massey-83} J.L. Massey, A Simplified Treatment of Wyner's Wiretap Channel, 21st Allerton Conf. on Comm., Control and Computing, pp. 268-276, Monticello, IL, Oct. 5-7, 1983.

\bibitem{Csiszar-78} I. Csiszar, J.Korner, Broadcast Channels with Confidential Messages, IEEE Trans. Inf. Theory, vol.~24, no.~3, pp. 339--348,  May 1978.


\bibitem{Leung-Yan-Cheong-78} S. K. Leung-Yan-Cheong and M. Hellman, The Gaussian Wire-Tap Channel, IEEE Trans. Info. Theory, v. 24, no. 4, pp. 451--456, July 1978.
\bibitem{Khisti-10-jul} A. Khisti, G. W. Wornell, “Secure transmission with multiple antennas - Part I: The MISOME wiretap channel,” IEEE Trans. Inf. Theory, vol. 56, no. 7, pp. 3088--3104, Jul. 2010.
\bibitem{Khisti-10-nov} A. Khisti, G. W. Wornell, “Secure transmission with multiple antennas - Part II: The MIMOME wiretap channel,” IEEE Trans. Inf. Theory, vol. 56, no. 11, pp. 5515--5532, Nov. 2010.
\bibitem{Oggier-11} F. Oggier, B. Hassibi, “The secrecy capacity of the MIMO wiretap channel,” IEEE Trans. Inf. Theory, vol. 57, no. 8, pp. 4961--4972, Aug. 2011.
\bibitem{Loyka-16b} S. Loyka, C.D. Charalambous, Optimal Signaling for Secure Communications Over Gaussian MIMO Wiretap Channels, IEEE Trans. Info. Theory, vol. 62, no. 12, pp. 7207-7215, Dec. 2016.

\bibitem{Dong-18} L. Dong, S. Loyka and Y. Li, The Secrecy Capacity of Gaussian MIMO Wiretap Channels Under Interference Constraints, IEEE Journal Sel. Areas Comm., vol. 36, no. 4, pp. 704-722, Apr. 2018.
\bibitem{Dong-20} L. Dong, S. Loyka and Y. Li, Algorithms for Globally-Optimal Secure Signaling Over Gaussian MIMO Wiretap Channels Under Interference Constraints, IEEE Trans. Signal Proc., vol. 68, pp. 4513--4528, Jul. 2020.

\bibitem{Mitrpant-06} C. Mitrpant et al, An Achievable Region for The Gaussian Wiretap Channel With Side Information, IEEE Trans. Info. Theory, vol. 52, no. 5, pp. 2181-2190, May 2006.
\bibitem{Merhav-21-Det} N. Merhav, Encoding Individual Source Sequences for the Wiretap Channel,  Entropy, 23(12):1694, Dec 2021.

\bibitem{Bastani-Parizi-17} M. Bastani Parizi, E. Telatar and N. Merhav, Exact Random Coding Secrecy Exponents for the Wiretap Channel, IEEE Trans. Info. Theory, vol. 63, no. 1, pp. 509--531, Jan. 2017.

\bibitem{Yang-19} W. Yang, R. F. Schaefer and H. V. Poor, Wiretap Channels: Nonasymptotic Fundamental Limits, IEEE Trans. Info. Theory, vol. 65, no. 7, pp. 4069--4093, Jul. 2019.


\bibitem{Liang-09} Y. Liang et al, Compound Wiretap Channels, EURASIP J. Wireless Commun. Netw., vol. 2009, Oct. 2009, Art. ID 142374.
\bibitem{Bjelakovic-13} I. Bjelakovic, H. Boche, J. Sommerfeld, Secrecy Results for Compound  Wiretap Channels, Probl. Inf. Transmission, vol. 49, no. 1, pp. 73--98, Mar. 2013.
\bibitem{Schaefer-15} R.F. Schaefer, S. Loyka, The Secrecy Capacity of Compound Gaussian MIMO Wiretap Channels, IEEE Trans. Info. Theory, v. 61, N. 10, pp. 5535 -- 5552, Oct. 2015.



\bibitem{Li-19} C. Li at al, Secrecy Capacity of Colored Gaussian Noise Channels With Feedback, IEEE Trans. Info. Theory, v. 65, no. 9, pp. 5771--5782, Sep. 2019.
\bibitem{Gunduz-08} D. Gunduz et al, Secret Communication With Feedback, Int. Symp. Info. Theory Appl., Auckland, New Zealand, Dec. 2008.
\bibitem{Schalkwijk-66} J. Schalkwijk, T. Kailath, A Coding Scheme for Additive Noise Channels With Feedback -- I: No Bandwidth Constraint, IEEE Trans. Info. Theory, vol. 12, no. 2, pp. 172--182, Apr. 1966.
\bibitem{Ardestanizadeh-09} E. Ardestanizadeh et al, Wiretap Channel With Secure Rate-Limited Feedback, IEEE Trans. Info. Theory, vol. 55, no. 12, pp. 5353--5361, Dec. 2009.

\bibitem{Keshet-08} G. Keshet, Y. Steinberg, N. Merhav, Channel Coding in the Presence of Side Information,  Foundations and Trends in Comm. Info. Theory, vol. 4, no. 6, pp. 445--586, June 2008.
\bibitem{Bross-20} S. I. Bross, A. Lapidoth, and G. Marti, Decoder-assisted communications over additive noise channels, IEEE Trans. Commun., vol. 68, no. 7, pp. 4150--4161, Jul. 2020.
\bibitem{Lapidoth-20} A. Lapidoth, G. Marti, Encoder-Assisted Communications Over Additive Noise Channels, IEEE Trans. Info. Theory, vol. 66, no. 11, pp. 6607--6616, Nov. 2020.
\bibitem{Marti-19} G. Marti, Channels With a Helper, M.S. Thesis, Signal Info. Process. Lab., ETH Zürich, Zurich, Switzerland, Sep. 2019.

\bibitem{Merhav-21} N. Merhav, On Error Exponents of Encoder-Assisted Communication Systems, IEEE Trans. Info. Theory, vol. 67, no. 11, pp. 7019--7029, Nov. 2021.


\bibitem{Fritschek-16} R. Fritschek and G. Wunder, Towards A Constant-Gap Sum-Capacity Result For The Gaussian Wiretap Channel With a Helper, Int. Symp. Info. Theory (ISIT), Jul. 2016, pp. 2978--2982.
\bibitem{Chen-20} J. Chen, C. Geng, Optimal Secure GDoF of Symmetric Gaussian Wiretap Channel With a Helper, IEEE Trans. Info. Theory,  v. 67, no. 4, pp. 2334--2352, Apr. 2021.


\bibitem{Cover-06} T.M. Cover, J.A. Thomas, Elements of Information Theory, Wiley, 2006.
\bibitem{Boyd-04} S. Boyd and L. Vandenberghe, Convex Optimization, Cambridge University Press, 2004.


\end{thebibliography}
\end{document}